\documentclass[11pt]{amsart}
\usepackage{amsmath}
\usepackage{amssymb}
\usepackage{amsthm}
\usepackage{mathrsfs}
\usepackage{comment}
\usepackage{hyperref}

\usepackage[all,cmtip]{xy}\usepackage{xcolor}
\usepackage{enumerate}
\usepackage{bm}
\usepackage{dsfont}
\usepackage{mathtools}
\usepackage[cal=euler]{mathalfa}
\usepackage[top=1in, bottom=1.25in, left=1.25in, right=1.25in]{geometry}
\usepackage{parskip}

\hypersetup{colorlinks=true,linkcolor=magenta,citecolor=blue}

\DeclareMathAlphabet{\mathpzc}{OT1}{pzc}{m}{it}

\theoremstyle{plain}
\newtheorem{theorem}{Theorem}[section]
\newtheorem*{theorem*}{Theorem}
\newtheorem*{maintheorem}{Theorem}
\newtheorem{lemma}[theorem]{Lemma}
\newtheorem*{claim*}{Claim}
\newtheorem{proposition}[theorem]{Proposition}

\newtheorem{corollary}[theorem]{Corollary}
\newtheorem{conjecture}[theorem]{Conjecture}

\theoremstyle{definition}
\newtheorem{definition}[theorem]{Definition}

\newtheorem{example}[theorem]{Example}

\newtheorem{remark}[theorem]{Remark}

\numberwithin{equation}{section}
\numberwithin{figure}{section}

\newcommand{\Petit}[2]{\mathbb{P}^{#1}_{#2}}
\newcommand{\Petita}[2]{\mathbb{P}^{#1}_{t^\de-#2}}

\newcommand{\ring}{S}
\newcommand{\field}{K}
\newcommand{\Ktsigma}{K[t;\sigma]}
\newcommand{\Stsigma}{S[t;\sigma]}

\newcommand{\fixedts}{S_0^\tau}
\newcommand{\scz}{\mathsf{z}}  
\DeclareMathOperator{\id}{id}
\newcommand{\os}{\mathbf{m}}  
\newcommand{\de}{\mathbf{n}}  

\newcommand{\N}{\mathbb{N}}
\newcommand{\F}{\mathbb{F}}
\newcommand{\Z}{\mathbb{Z}} 

\DeclareMathOperator{\Gal}{Gal}
\DeclareMathOperator{\Aut}{Aut}



\begin{document}
\title[Isometry and equivalence]{When isometry and equivalence for skew constacyclic codes coincide}

\author{Monica Nevins}
\address{Department of Mathematics and Statistics, University of Ottawa, Ottawa, Canada K1N 6N5}
\email{mnevins@uottawa.ca}
\thanks{Corresponding author: Susanne Pumpl\"un. Email: \href{mailto:Susanne.Pumpluen@nottingham.ac.uk}{Susanne.Pumpluen@nottingham.ac.uk}}

\author{Susanne Pumpl\"un}
\address{School of Mathematical Sciences, University of Nottingham,
Nottingham NG7 2RD}
\email{Susanne.Pumpluen@nottingham.ac.uk}



\date{\today}

\begin{abstract}
We work in the  setting of
linear skew constacyclic codes over a commutative base ring $S$.
We show that the notions of $(\de,\sigma)$-isometry and $(\de,\sigma)$-equivalence introduced by Ou-azzou \emph{et al}
coincide for most skew  $(\sigma,a)$-constacyclic codes of length $\de$, and moreover identify the classes for which the two notions can differ.  To prove these results, we first determine all Hamming-weight preserving isomorphisms between their ambient Petit rings which extend some automorphism $\tau$ of $S$ that commutes with $\sigma$.  We then prove that when those ambient rings are not associative, these isomorphisms must have degree one, and give examples of non-degree one isomorphisms in the associative case.   As a consequence, we  provide new definitions of equivalence and isometry that for skew constacyclic codes over finite rings  exactly capture all  Hamming-weight preserving isomorphisms between their ambient rings,  leading to tighter classifications.
\end{abstract}

\maketitle

\section{Introduction}

Let $\ring$ be a unital commutative associative ring, $a\in \ring$ invertible and $\sigma\in \Aut(S)$. Then a linear code $C\subset S^{\mathbf{n}}$ is called a skew  $(\sigma, a)$-constacyclic code, if for every  codeword $(c_0,\dots,c_{\mathbf{n}-1})\in  C$, the ``shifted'' codeword  $(\sigma(c_{\mathbf{n}-1})a,\sigma(c_0),\dots,\sigma(c_{\mathbf{n}-2}))$ also lies in $C$. More generally, a linear code is a skew polycyclic code, or more precisely called a  skew  $(f,\sigma)$-polycyclic code, if every codeword satisfies a similar but more elaborate shift, this time determined by the coefficients of a pre-chosen reducible skew polynomial $f(t)\in \Stsigma$.

\emph{Petit rings} are nonassociative unital rings that can be considered as the ambient rings of skew $(f,\sigma)$-polycyclic codes,  as there is a one-to-one correspondence between a skew polycyclic code and a principal left ideal in a Petit ring which is generated by a monic right divisor of the skew polynomial $f(t)$ \cite{Pumpluen2017, Pumpluen2025}.
 For instance, skew $(\sigma, a)$-constacyclic codes  correspond to the left principal ideals in their nonassociative ambient Petit ring $\Petita{S}{a}$ whose underlying module is $\Stsigma/\Stsigma(t^\de-a)$
 (notation as in Section~\ref{S:quotient}).  Left multiplication of any element in a left principal ideal of the ambient (albeit potentially nonassociative) Petit ring  by the indeterminate $t$ then creates the shift employed in the definition of a skew polycyclic code.  Thus while skew $(\sigma, a)$-constacyclic codes are often viewed  as $\Stsigma$-submodules of $\Stsigma/\Stsigma (t^\de-a)$, using the perspective of  nonassociative Petit rings significantly enriches the algebraic tools available for classification, as this paper illustrates.
In contrast, the widely used module theoretical viewpoint introduced in \cite{BoulagouazLeroy2013}  is less natural for studying the isometries between skew polycyclic codes.  Namely, the multiplication in a Petit ring immediately imposes strong restrictions on their isomorphisms and consequently, the isometries. Most importantly perhaps, by studying ambient Petit rings, our approach immediately removes the \emph{artificial} restriction that the coefficient ring of the code needs to be fixed point-wise under isometries.
This perspective broadens our horizons and shows us how to systematically find new well-behaved linear codes by looking at ideals of other nonassociative algebras, as done for example in \cite{PumpluenMenchetti2025}.

It is also important to look beyond the case that $S$ is a finite field. Skew polycyclic codes over finite commutative rings (and especially non-chain rings)  have recently gained considerable attention in such works as \cite{MohanGowdhamanRadhakrishnanetal2026},\cite{JitmanLingU2012},\cite{FotueTabueMartinezMoroBlackford2020}.  Therefore we will develop our framework for any unital commutative ring $S$.

 The \emph{Hamming-weight} of a \emph{codeword} $c=(c_0,\dots,c_{\de-1})$ is defined as the number of nonzero coefficients of $c$.  The Hamming-weight of the corresponding polynomial $p(t)=c_0+c_1t+\cdots + c_{\de-1}t^{\de-1}$ is thus the number of nonzero coefficients of $p(t)$.
Two skew polycyclic  codes are called \emph{isometric}
if the restriction of a \emph{monomial} isomorphism between their ambient Petit rings maps one code to the other.  These are exactly the isomorphisms, denoted here $G_{\tau,\alpha,k}$, which map $t$ to $\alpha t^k$ for some positive integer $k$ and invertible $\alpha\in S$, and restrict to some $\tau\in \Aut(S)$ commuting with $\sigma$.
  Such an isomorphism between Petit rings, also called an \emph{isometry}, preserves the Hamming-weight of a polynomial and,
thus gives a correspondence between the sets of  codes (\emph{i.e.}, the sets of principal left ideals of the two isometric ambient rings) that preserves code rate and Hamming distance as well as  length and dimension  \cite{Pumpluen2025}.

 There has been significant work to understand the isometries of skew polycyclic codes. In \cite{OuazzouHorlemannAydin2025, OuazzouNajmeddineAydin2025}, the authors consider a subclass of such isometries (namely, those monomial isomorphisms that act as the identity on the finite field $\mathbb{F}_q$) to propose the notions of $(\de,\sigma)$-equivalence and $(\de,\sigma)$-isometry between skew constacyclic codes (and of $(\de,\sigma)$-equivalence for skew polycyclic codes) over $\mathbb{F}_q$, generalizing \cite{ChenDinhLiu2014}.
The same subclass of isometries was used, but between the algebras  $\Petita{\F_q}{a}$ and $\Petita{\F_q}{b}$, to define equivalences of skew constacyclic codes
in \cite{BoulanouarBatoulBoucher2021}
and later again in  \cite{LobilloMunoz2025}.  However, this latter paper only considered a restricted case: where $\de$ is a multiple of the order of $\sigma$
and $a$ fixed by $\sigma$; in this case the ambient ring is associative.

In this paper we generalize  the notions of $(\de,\sigma)$-equivalence and $(\de,\sigma)$-isometry to equivalence and isometry, allowing for $\tau$ to be nontrivial in our definitions. We prove that the notions of  equivalence and isometry, thus in particular $(\de,\sigma)$-equivalence and $(\de,\sigma)$-isometry, coincide for all classes of skew constacyclic codes over $S$ whose ambient Petit rings are not associative. 
This result was only previously known for when the order $\os$ of $\sigma$  satisfied $\os\geq \de-1$ and when $\ring=K$ was a field \cite{BrownPumpluen2018}.

 Readers interested in codes can simply skip the technical Sections~\ref{S:3}, \ref{S:4} and \ref{S:nonmonomial} that prove the requisite properties of  isometries of Petit rings, and proceed to Section~\ref{S:applications} directly, which applies these  results to skew constacyclic codes.

The motivation for the present work was
two-fold: first, we wanted  to understand what turned out to be an error in
\cite[Theorem 4]{OuazzouNajmeddineAydin2025},  which suggested the existence of isometries of higher degree for all skew constacyclic codes.  Searching for these led us to discover a rather startling result.

To state it succinctly, let $S=K$ be a finite field.  The Petit ring   $\Petita{K}{a}$ is said to be  \emph{proper nonassociative} when it is not an associative ring; by \cite[(15)]{Petit1967} this occurs exactly when $\sigma(a)\neq a$ or $\os\nmid \de$.

\begin{maintheorem}[Corollary~\ref{C:onlyweightone}]
Let $\field$ be a field with $\sigma\in \Aut(\field)$ of finite order $\os$, and let $a_1,a_2\in \field^\times$.  If the Petit rings $\Petita{K}{a_i}$
are proper nonassociative, then all nonzero Hamming-weight preserving homomorphisms $G$ between them such that $G|_{\field}$ commutes with $\sigma$ are monomial of degree one, that is, satisfy $G(t)=\alpha t$ for some $\alpha \in \field^\times$.
\end{maintheorem}

In particular, the theorem implies that these rings admit no nonzero Hamming-weight preserving homomorphisms
of the form $G(t)=\alpha t^k$ for any $\alpha\in \field^\times$ and $1<k<\de$.

Our approach is novel, and applies in an  important greater generality, where the field $K$ is replaced by any commutative unital (associative) ring $\ring$ (such as a Galois ring, or a finite chain ring, for example).  The first step is to characterize the power-associative monomials $\alpha t^k \in \Petita{S}{b}$,
where $\alpha,b \in \ring^\times$. The answer is quite elegant and presented in Section~\ref{S:3}.
\newtheorem*{T:power}{Theorem \ref{T:powerassociative}}
\begin{T:power}
    Suppose $\alpha \in \ring$ is not a zero divisor and $1\leq k < \de$.  The monomial $\alpha t^k \in \Petita{S}{b}$
    is power-associative if and only if
    \begin{equation}\label{E:introassoc}
    \alpha b = \sigma^{\de}(\alpha)\sigma^k(b) \in\Petita{S}{b}.
    \end{equation}
\end{T:power}

The proof turns on a careful combinatorial argument which occupies most of Section~\ref{S:3}.

As a corollary, we give a classification of all \emph{monomial homomorphisms} $G:\Stsigma\to \Petita{S}{b} $
of the form $G|_S=\tau$ and $G(t)=\alpha t^k$ for some $\tau\in \Aut(S)$ commuting with $\sigma$, $\alpha\in S^\times$ and $1\leq k <\de$
(Corollary~\ref{C:StsigmatoSb}).
In particular, if such a homomorphism with $k>1$ exists,
then this entails $\os\mid \de$, $k\equiv 1 \mod \os$ and thus $\sigma(b)=b$, whence $\Petita{S}{b}$
is associative.

 We then leverage this to prove in Corollary~\ref{C:onlyweightone} that if $\Petita{S}{b}$
 is proper nonassociative, then in fact there are no nonzero monomial homomorphisms $G: \Petita{S}{a}\to \Petita{S}{b}
 $ of degree other than one. Thus, in particular,  $(\de,\sigma)$-isometry and $(\de,\sigma)$-equivalence as introduced in \cite{OuazzouNajmeddineAydin2025} coincide for  classes of skew constacyclic codes that arise from proper nonassociative Petit ambient rings.

We then explain how our results lead to the failure of \cite[Theorem 4]{OuazzouNajmeddineAydin2025}  in Section~\ref{S:applications}, providing explicit counterexamples.  We also recalculate the number of $(\de,\sigma)$-equivalence classes of ambient Petit algebras $\Petita{K}{a}$ of skew constacyclic codes with, for $\field$  a finite field, in Theorem~\ref{c:Ouazzoufinite}.

In Section~\ref{S:applications} we define more general notions of equivalence and isometry that perfectly capture the set of isomorphisms that will preserve Hamming distance.  This offers a tighter classification of equivalence and isometry classes of codes (see Definition~\ref{D:isometric}).   The advantage of this perspective is clearly demonstrated in
the many examples in Section~\ref{S:applications}.
For associative algebras, we can prove a variation of \cite[Theorem 4]{OuazzouNajmeddineAydin2025}  (Theorem \ref{t:Ouazzou}).
  In particular, the existence of a monomial isomorphism of degree $k>1$ between $\Petita{S}{a}$ and $\Petita{S}{b}$
  in the associative case gives rise to monomial isomorphisms of degree one from $\Petita{S}{a}$ to $\Petita{S}{b^k}$.

Our second motivation for this work is to eventually extend our parametrization of proper nonassociative Petit division algebras begun in \cite{NevinsPumpluen2025} (and recently continued in \cite{NevinsPumpluen2026b}).
 In that paper, we solved the case that $\de=\os$ is prime and $f(t)=t^{\de}-a$.
 These division algebras, in turn, are in one one correspondence with \emph{maximum rank distance codes}.   Any parametrization will thus yield parametrizations of the corresponding codes. For $\de=\os$, the division algebras $\Petita{K}{a}$ have just been used in a variation of Learning with Errors, a lattice based approach to postquantum cryptography  \cite{MendelsohnLing2025}.

Crucially, to extend the result in \cite{NevinsPumpluen2025} to any  $\de$ and $\os$ one needs to understand all possible isomorphisms between two division algebras $\Petita{K}{a}$ and $\Petita{K}{b}$, not only the monomial ones.
Thus in Section~\ref{S:nonmonomial}, for any integral domain  $\ring$, we prove that  polynomial, \emph{i.e.},  non monomial, morphisms between proper nonassociative algebras of the form $\Petita{S}{a}$
do not occur, subject to a technical hypothesis we denote \eqref{stareq}.
Consequently, it follows from  Theorem  \ref{P:polyndivm} that for integral domains $\ring$, $a,b\in \ring^\times$, and $\os \nmid \de$, every nonzero homomorphism
 $G:\Petita{S}{a}\to \Petita{S}{b}
 $  whose restriction to $\ring$ is given by some automorphism $\tau$ commuting with $\sigma$, must be monomial of degree one (Corollary \ref{c:new}).

Skew constacyclic and skew polycyclic codes over finite rings and mixed alphabets are currently among the most studied kinds of linear codes.
Our current work serves to improve the de-duplication of codes by
fully determining the structure of \emph{all} possible
isometries between them.
We subsequently employed some of the results from this paper to obtain a first tight classification of families of skew constacyclic codes over finite fields in \cite{NevinsPumpluen2026b}.

We point out that our results for isometries between the ambient Petit algebras also have consequences for skew polycyclic codes over finite fields with the rank metric, where we need $\alpha\in K'$ for a suitable subfield $K'$ of $K$  we consider the rank metric over. In particular,  they will  help understand the different  MRD codes obtained from these \cite[Proposition 3.13]{OuazzouHorlemannAydin2025}.

\section{Preliminaries}\label{S:prel}

Let $\ring$ be a commutative, associative unital ring, $S^\times$ its set of invertible elements, and suppose $\sigma\in \Aut(\ring)$ has finite order $\os$.

\begin{remark}
    Important examples in applications to coding theory are $\ring=\F_q$, a finite field of order $q=p^n$ for some prime $p$ and $n\in \mathbb{N}$, a finite chain ring (also called a finite local ring) such as $\Z/p^n\Z[\alpha]$ for some algebraic integer $\alpha$, or $\ring$ a Galois ring;
    see for example \cite[\S5.1]{Pumpluen2017}.   In every \emph{finite} ring, every non-invertible element is a zero divisor; in infinite rings, a nonzero divisor is cancellable but not necessarily invertible.
\end{remark}

 Write $\ring_0=\{s\in \ring: \sigma(s)=s\}$ for the subring fixed by $\sigma$.
If $\tau\in \Aut(\ring)$ is an automorphism commuting with $\sigma$ then we write $\fixedts=\{s\in S_0: \tau(s)=s\}$ for the common fixed subring under $\tau$ and $\sigma$.

The ring $R=\Stsigma$ is a unital associative noncommutative skew polynomial ring, with addition the usual addition of polynomials, and with multiplication defined by the relation
$$
t^ja=\sigma^j(a)t^j
$$
for all $a\in \ring$ and all $j\in \N$.

Let  $f(t)=t^{\de}-f_0(t)\in \ring[t;\sigma]$ be a monic polynomial of degree $\de\geq 1$.
For all $g(t)\in  \ring[t;\sigma]$ of degree $\ell\geq \de$, there exist uniquely
determined $r(t),q(t)\in  \Stsigma$ with
 ${\rm deg}(r)<{\rm deg}(f)$, such that
$g(t)=q(t)f(t)+r(t)$, e.g. see \cite[Proposition 1]{Pumpluen2017}; we call this \emph{right} division by $f(t)$. Then $\Stsigma f(t)$ is a left ideal of $S[t;\sigma]$.

One may form the familiar $\Stsigma$-module
$
S[t;\sigma]/S[t;\sigma]f(t).
$
Since the remainders modulo right divison by $f(t)$ are uniquely determined, we may identify the additive left cosets of the left ideal $\Stsigma f(t)$ with the set
$$
\Stsigma_\de = \left\{ \sum_{i=0}^{\de-1} b_it^i : b_i\in \ring\right\} .
$$
While this is a module over $\Stsigma$,  it is not \emph{a priori} clear that there is always a well-defined ring multiplication on this set of representatives $\Stsigma_\de$
 as well.

 In the special case that  $\Stsigma f(t)$ is a two-sided ideal  then we call $f(t)\in R$  \emph{two-sided}, and $\Stsigma/\Stsigma f(t)$ is an associative quotient ring.  For example, this holds for all $f(t)$ in the center of $\Stsigma$. In this case, the product of two cosets is well-defined and for $g,h\in \Stsigma_\de$, it simplifies to the rule $g\circ h=gh \,\,{\rm mod}_r f$, where $g \,{\rm mod}_r f$ denotes the remainder of right division of $g$ by $f$.  

 \subsection{Petit rings are not quotient rings}\label{S:quotient}
Even when $f(t)=t^\de - f_0(t)\in \Stsigma $ is not two-sided, we have a well-defined multiplication map  $\Stsigma_\de \times \Stsigma_\de\to \Stsigma_\de$ via
$$
g\circ h=gh \,\,{\rm mod}_r f.
$$
In particular,  for all $0\leq i,j<\de$ we have
\begin{equation}\label{Equation:tde+1simplified}
t^it^j = \begin{cases}
    t^{i+j} & \text{if $i+j<\de$;}\\
    t^{i+j-\de}f_0(t) & \text{if $i+j\geq \de$.}
\end{cases}
\end{equation}
This endows $\Stsigma_\de$ with the structure of a unital \emph{nonassociative} ring. We call the unital ring $$\Petit{S}{f}:=(\Stsigma_\de,\circ)$$
 defined in this way a \emph{Petit ring} and note that it can also be viewed as an algebra over $S$ (or  over any subring of $S$). Note that we use the superscript $S$ to identify the base ring $S$; when $S$ is a field we use the letter $K$ as in $\Petit{K}{f}$.
 The Petit ring can be viewed as the canonical generalization of the classical quotient ring to this setting.

We emphasize this multiplication is not associative in general: by \cite{Petit1967}, $\Petit{S}{f}$
is an associative ring if and only if $\Stsigma f(t)$ is a two-sided ideal in $\Stsigma$.
As an important example:  if $f(t)=t^\de-a$, then by   $\Petita{S}{a}$
is associative if and only the order $\os$ of $\sigma$ divides $\de$ and $a\in S_0$   (\emph{e.g.}, \cite[(15)]{Petit1967}).

\subsection{The element \texorpdfstring{$t$}{t} is not always power associative}\label{S:associativity}

We are particularly interested in the special case that $f(t) = t^{\de}-a$ for some $a\in \ring^\times$.
One must take care to simplify products only with \emph{right} division by $t^{\de}-a$.
  A key observation from \eqref{Equation:tde+1simplified} is that in $\Petita{S}{a}$
  we have
$$
t^{\de} t = a t \quad \text{whereas} \quad t t^{\de} = t a = \sigma(a) t.
$$
This example shows that when $a\not\in S_0$, we have
$$t^{\de} t\not= t t^{\de}$$
even though $t^it^j =  t^{i+j} $ for all $i$,$j$ such that $i+j<\de$, by \eqref{Equation:tde+1simplified} 
(with $f_0(t)=a$).

In general, the element $t$ is power-associative in the ring $\Petit{S}{f}$  (which is equivalent to saying all powers of $t$ commute), if and only if $t^{\de} t= t t^{\de}$, which by \eqref{Equation:tde+1simplified}  is equivalent to $f(t)t\in \Stsigma f(t)$; this, in turn, is equivalent to $t$ lying in the \emph{right nucleus} of $\Petit{S}{f}$
  \cite[(5), p. 13-06]{Petit1967}.

 Our next concern is to find the class of homomorphisms between Petit rings  that preserves the Hamming-weight of a polynomial, and thus induces exactly  the  Hamming-weight preserving homomorphisms of skew polycyclic codes (Section~\ref{SectionSkewpolycycliccodes}).
 
\subsection{Monomial homomorphisms between Petit rings}
We require some additional notation.   For any $i\in \mathbb{N}$, $\sigma\in \Aut(\ring)$ and $\beta\in \ring$, write
$$
N_i^{\sigma}(\beta) = \prod_{j=0}^{i-1}\sigma^j(\beta).
$$
Note that $\beta \sigma(N_i^\sigma(\beta))=N_{i+1}^\sigma(\beta)$ and thus
\begin{equation}\label{E:Normrelation}
N_{i+j}^\sigma(\beta)= N_{i}^\sigma(\beta) \cdot \sigma^{i}(N_j^\sigma(\beta))
\end{equation}
for all $i,j\geq 0$ (as also observed in \cite[Proposition 2.1]{CherchemLeroy2016} for the case of  finite fields).
If $\ring=\field$ is a field, $\field/\field_0$ a cyclic Galois extension, and $\Gal(\field/\field_0)=\langle \sigma \rangle$, where $\sigma$ has order $\os$, then  $N_{\os}^\sigma$ is simply the norm map $N_{\field/\field_0}$ corresponding to $\field/\field_0$.  Even when $S$ is not a field, when $\sigma$ has order $\os$ then $N_{\os}^\sigma(\beta)\in \ring_0$ for all $\beta\in \ring$ so we distinguish this case with the notation
$N_{\os}^\sigma=N_{\ring/\ring_0}$.

\begin{definition}\label{D:Gtaualpha}
    Let $f,h\in \Stsigma$ be monic polynomials of degree $\de$.  Suppose $\tau\in \Aut(\ring)$, $\alpha \in \ring$ is nonzero and $k\in \mathbb{N}$.  If there exists a ring homomorphism
    $$G:\Petit{S}{f}\to\Petit{S}{h}$$     that
      is  defined via $G|_{S}=\tau$ and $G(t)=\alpha t^k$,    that is,
\begin{equation}\label{E:formulaG(t)2}
G\left( \sum_{i=0}^{\de-1}a_it^i \right) = \sum_{i=0}^{\de-1}\tau(a_i) (\alpha t^k)^i.
\end{equation}
     then we call it a \emph{monomial homomorphism  of degree $k$} and write $G=G_{\tau,\alpha,k}$.  When $k=1$, we write $G_{\tau,\alpha}$ in place of $G_{\tau,\alpha,1}$, in which case we have simply
$$    G_{\tau,\alpha}\left(\sum_{i=0}^{\de-1}a_it^i\right) = \sum_{i=0}^{\de-1}\tau(a_i)N_i^\sigma(\alpha)t^i.    $$
A monomial homomorphism that is an isomorphism is called a \emph{monomial isomorphism}.
\end{definition}

By abuse of notation, we will later also denote by $G_{\tau,\alpha,k}$ a ring homomorphism $G: \Stsigma \to \Stsigma$ or $G: \Stsigma\to \Petit{S}{h}$ that is generated by $G|_{S}=\tau$ and $G(t)=\alpha t^k$.

As we shall discuss in Section~\ref{S:3}, there is significant subtlety in the well-definedness of the term $(\alpha t^k)^i$ appearing in \eqref{E:formulaG(t)2} in the nonassociative case (\emph{cf}. power-associativity).

There are several prior results showing that, under certain restrictions, all isomorphisms between Petit rings are monomial of degree one.
The following results were originally proved only for $\ring_0$-isomorphisms for the case that $\ring$ is a field in \cite[Theorems 28, 29]{BrownPumpluen2018}, and generalized to the present setting in  \cite{Pumpluen2025}.
To state them, set
 $$
    f(t) =t^{\de}- \sum_{i=0}^{\de-1} a_it^i, \quad h(t) = t^{\de}-\sum_{i=0}^{\de-1}b_it^i.
    $$

\begin{theorem}\label{T:general_isomorphism_theorem2} \cite[Theorem 4.1]{Pumpluen2025}
Let $\alpha\in S^\times$ and assume that $\tau \in {\rm Aut}(S)$ commutes with $\sigma$. Then
   $$ G_{\tau,\alpha}:  \Petit{S}{f} \to \Petit{S}{h}
    $$
    is an  isomorphism  between nonassociative unital rings that restricts to the identity on  $S_0^\tau$
    if and only if
\begin{equation}  \label{equ:necII}
\tau(a_i) =
N_{\de-i}^{\sigma}(\sigma^{i} (\alpha))b_i
\end{equation}
 for all $i \in\{ 0, \ldots, \de-1\}$.
\end{theorem}

When $S=K$ is a field, $\Petit{K}{f}$ and $\Petit{K}{h}$ are proper nonassociative Petit rings, and $\Aut(\field)$ is abelian, we can say even more.

\begin{theorem}\label{T:BPonlyonesareGtaualpha} \cite[Theorem 4.2]{Pumpluen2025}
  Let $\os \geq \de-1$ and assume that $f$ and $h$  do not generate two-sided ideals in $\Ktsigma$.  Suppose $\Aut(\field)$ is abelian.  For any subfield $F\subset \field_0$, the \emph{only} ring isomorphisms
  $$
  G:\Petit{K}{f}\to \Petit{K}{h}
  $$
  that restrict to the identity on $F$
  are of the form $G=G_{\tau,\alpha}$ such that $\tau\in \Aut(\field)$,   $F\subset \field_0^\tau$, and $\alpha \in \field^\times$ satisfying Theorem~\ref{T:general_isomorphism_theorem2}.
\end{theorem}

In this paper, we expand the set of cases for which it is known that \emph{all} isomorphisms of Petit rings are monomial  of type $G_{\tau,\alpha}$, when  $\tau \in {\rm Aut}(S)$ commutes with $\sigma$.  While our results are valid for all $\os, \de$, their novelty  is in the open case that $\os < \de -1$, and in that we do not require that $\ring$ is a field.

 \subsection{Skew polycyclic and constacyclic codes} \label{SectionSkewpolycycliccodes}

 A \emph{linear code  of length $\mathbf{n}$ over $S$} is a  free submodule of the $S$-module $S^\mathbf{n}$.
Let $$f(t)=t^{\mathbf{n}}-\sum_{i=0}^{\mathbf{n}-1}a_it^i\in \Stsigma$$
be a reducible monic polynomial of degree $\mathbf{n}>1$.
 Then a linear code $C\subset S^{\mathbf{n}}$ is called a   \emph{skew  $(f,\sigma)$-polycyclic code}, or \emph{skew polycyclic code} for short, if for each   $(c_0,c_1,\ldots, c_{\mathbf{n}-1})\in C$,  also
$$
\left( 0, \sigma(c_0),\sigma(c_1),\ldots, \sigma(c_{\mathbf{n}-2})\right) + \sigma(c_{\mathbf{n}-1}) (a_0,a_1,\ldots, a_{\mathbf{n}-1}) \in C.
$$
Note that
$$
\left( 0, \sigma(c_0),\sigma(c_1),\ldots, \sigma(c_{\mathbf{n}-2})\right) +  (\sigma(c_{\mathbf{n}-1}) a_0,\sigma(c_{\mathbf{n}-1})a_1,\ldots, \sigma(c_{\mathbf{n}-1})a_{\mathbf{n}-1})$$
$$
  = ( \sigma(c_{\mathbf{n}-1}) a_0, \sigma(c_0)+\sigma(c_{\mathbf{n}-1})a_1,\sigma(c_1)+\sigma(c_{\mathbf{n}-2})a_2,\ldots, \sigma(c_{\mathbf{n}-2})+\sigma(c_{\mathbf{n}-1})a_{\mathbf{n}-1}).
$$
 A  skew  $(f,\sigma)$-polycyclic code $C\subset S^{\mathbf{n}}$ is called a
 \emph{skew  $(\sigma, a)$-constacyclic code} if  $f(t)=t^{\mathbf{n}}-a$ for an invertible $a\in S^\times$, that is
   $$(c_0,\dots,c_{\mathbf{n}-1})\in  C\Rightarrow (\sigma(c_{\mathbf{n}-1})a,\sigma(c_0),\dots,\sigma(c_{\mathbf{n}-2}))\in  C.$$
   As proven in \cite{Pumpluen2025}, the bijective map $\Phi:S^{\mathbf{n}}\rightarrow \Stsigma_{\de}$,
$$(c_0,c_1,\dots,c_{\mathbf{n}-1})\mapsto \sum_{i=0}^{\mathbf{n}-1}c_it^i,$$ sends a skew $(f,\sigma)$-polycylic code $C$  of length $\mathbf{n}$ onto a principal  left ideal $C(t)=\Phi(C)$ of the nonassociative Petit ring $\Petit{S}{f}$ (the \emph{ambient ring} of the code).  Such an ideal is generated by a monic right divisor of $f$.  That is, for a skew $(f,\sigma)$-polycylic code $C$  of length $\mathbf{n}$ over $S$ we denote by $C(t)$ the set of skew-polynomials
 $c(t)=\sum_{i=0}^{\mathbf{n}-1}c_it^i$ associated to the codewords $(c_0,\dots,c_{\mathbf{n}-1})\in C$, and these form a principal left ideal in  $\Petit{S}{f}$.
  Note that the left multiplication by $t$ with an element of a left principal ideal in the ambient ring $\Petit{S}{f}$
  then produces exactly the shift that is used in the definition of a skew $(f,\sigma)$-polycyclic code.

It follows that any monomial isomorphism $G_{\alpha,\tau,k}$ of the ambient ring thus induces a Ham\-ming-weight preserving isomorphism between codes. To determine when such isomorphisms exist, we must first classify general monomial \emph{homomorphisms} from skew polynomial rings to Petit rings in Section~\ref{S:3}, and then between Petit rings in Section~\ref{S:4}.

\section{Monomial homomorphisms from \texorpdfstring{$\Stsigma$}{STsigma} to a Petit ring} \label{S:3}

Homomorphisms between skew polynomial rings with $\ring$ a division algebra were studied in \cite{LamLeroy1992}. Surjective homomorphisms between $\Stsigma$ for general rings $\ring$ were considered in \cite{Rimmer1978}.

Let $\ring$ be a commutative unital ring.
We begin with a general restriction on monomial homomorphisms $G:\Stsigma\to\Stsigma$.

\begin{lemma}\label{L:Ktsigmatoitself}
    Let $\tau\in \Aut(\ring)$, $k\in \N$,  and $\alpha\in \ring$ such that $\alpha$ is not a zero divisor.  Then the map generated by $G(c)=\tau(c)$ for all $c\in \ring$ and $G(t)=\alpha t^k$ defines a homomorphism of rings from $\Stsigma$ to itself if and only if $\tau\sigma\tau^{-1}=\sigma^k$.  In particular, if $\tau$ and $\sigma$ commute then this occurs if and only if $k\equiv 1 \mod \os$.

    Moreover, if $\tau$ and $\sigma$ commute, then $G$ is an isomorphism if and only if $k=1$ and $\alpha\in \ring^\times$.
\end{lemma}

\begin{proof}
The map $G:\Stsigma \to \Stsigma$ generated by $G(c)=\tau(c)$ for all $c\in \ring$ and $G(t)=\alpha t^k$ is given by
$$
G\left(\sum_{i=0}^{\ell} a_i t^i\right) = \sum_{i=0}^\ell \tau(a_i) N_i^{\sigma^k}(\alpha) t^{ik},
$$
for any $\ell\in \mathbb{N}$ and $a_i\in \ring$.
This map is $\ring_0^\tau$-linear and will be multiplicative if and only if  $G(t)G(c)=G(\sigma(c))G(t)$ for all $c\in \ring$.
 We compute
$$
G(t)G(c)=\alpha t^k \tau(c) = \alpha \sigma^k(\tau(c)) t^k
$$
and
$$
G(\sigma(c))G(t)=\tau(\sigma(c))\alpha t^k.
$$
Since $\alpha$ is not a zero-divisor,  this holds for all $c\in \ring$ if and only if $\tau\sigma=\sigma^k\tau$.  In particular, if $\sigma$ and $\tau$ commute, we must have $\sigma=\sigma^k$, or  $k\equiv 1 \mod \os$.

Suppose now that $\sigma$ and $\tau$ commute. If $G$ is to be an isomorphism, it must be surjective. For every nonzero $f\in  \Stsigma$, we have $\deg (G(t))=k\deg(f)$, hence all non-zero skew polynomials in the image of $G$ will have degree divisible by $k$, in particular so will $t$ which forces $k=1$.
 Moreover, if $\alpha$ were not invertible, then $t^k$ could not lie in the image of $G$. This implies that $\alpha\in S^\times$ and $k=1$.

Conversely,  if $\sigma$ and $\tau$ commute then every $G$ with $\alpha\in S^\times$ and $k=1$ is clearly bijective.
\end{proof}

A special case of this result (where $\tau=\id$) has appeared as early as \cite[Theorem 3]{Rimmer1978}.

From now on, we assume that $\tau\in \Aut(\ring)$ commutes with $\sigma$, and that $\alpha \in \ring^\times$.
In this setting, the resulting homomorphism $G$ is an isomorphism if and only if $k=1$.  Then for any $a\in \ring^\times$ and $\de\in \N$ we have
$$
G(t^{\de}-a) = N^\sigma_{\de}(\alpha)t^{\de}-\tau(a) = N^\sigma_{\de}(\alpha)(t^{\de} - N^\sigma_{\de}(\alpha^{-1})\tau(a)).
$$
Thus $G$ induces a ring isomorphism
$$
G : \Petita{S}{a} \to \Petita{S}{b} \quad \text{where $\tau(a)= bN^\sigma_{\de}(\alpha)$,}
$$
and this isomorphism fixes $\fixedts\subset S$ pointwise.

In this section we show that these are \emph{all} the  Hamming-weight preserving homomorphisms from $\Stsigma$ to $\Petita{S}{b}$.

The first key obstruction for a map into $\Petita{S}{b}$ to be a homomorphism is the power-associativity of the term $\alpha t^k$.  We first require a technical lemma that will be used  in Theorem~\ref{T:powerassociative} to reduce the question of associativity of powers of $\alpha t^k$ in $\Petita{S}{b}$ to a straightforward relation between $\alpha$ and $b$.  Write $\lfloor x \rfloor$ for the greatest integer that is less than or equal to $x$.

\begin{lemma}\label{L:keystep}
    Suppose $\de,k\in \mathbb{N}$ and $\alpha,b\in \ring$ are not zero divisors.  Suppose
\begin{equation}\label{E:conditiononassociativity}
    \alpha b = \sigma^{\de}(\alpha) \sigma^k(b).
    \end{equation}
    Then for all $s,\ell\in \mathbb{N}$ we have
    \begin{equation}\label{E:keyequality}
\sigma^{[sk]_{\de}}
(N_{\ell}^{\sigma^k}(\alpha))
\prod_{j=1}^z\sigma^{sk-j{\de}}(b) =
\sigma^{sk}
(N_{\ell}^{\sigma^k}(\alpha))
\prod_{j=1}^z\sigma^{sk-j{\de}}
(\sigma^{\ell k}(b)),
    \end{equation}
    where $z= \lfloor \frac{sk}{{\de}}\rfloor$ and $[sk]_{\de}=sk-z{\de}$ is the residue mod ${\de}$ of $sk$.
\end{lemma}

\begin{proof}
   Since $\alpha,b$ are not zero divisors the same is true of every image under a power of $\sigma$, and we may work in the total ring of fractions of $\ring$ obtained by localizing on the set of non zero divisors. Then the hypothesis \eqref{E:conditiononassociativity} can be rephrased as
   $$
    \frac{\alpha}{\sigma^{\de}(\alpha)}=\frac{\sigma^k(b)}{b}.
    $$
It follows that for any positive integer $z$ we have
$$
\frac{\alpha}{\sigma^{z{\de}}(\alpha)} = \prod_{i=0}^{z-1} \sigma^{i{\de}}\left(\frac{\alpha}{\sigma^{\de}(\alpha)}\right)= \prod_{i=0}^{z-1} \sigma^{i{\de}}\left(\frac{\sigma^k(b)}{b}\right)= \prod_{j=1}^{z}\sigma^{z{\de}-j{\de}}\left(\frac{\sigma^k(b)}{b}\right).
$$
Now let $\ell \in \mathbb{N}$.  Applying the function $N_\ell^{\sigma^k}$ to both sides and then simplifying yields
$$
  \prod_{i=0}^{\ell-1}\sigma^{ik}\left( \prod_{j=1}^{z}\sigma^{z{\de}-j{\de}}\left(\frac{\sigma^k(b)}{b}\right) \right)
  = \prod_{j=1}^{z}\sigma^{z{\de}-j{\de}}\left(\prod_{i=0}^{\ell-1}\sigma^{ik}\left(\frac{\sigma^k(b)}{b}\right) \right) = \prod_{j=1}^{z}\sigma^{z{\de}-j{\de}}\left(\frac{\sigma^{\ell k}(b)}{b}\right).
$$
That is, \eqref{E:conditiononassociativity} implies that for all $z,\ell\in \mathbb{N}$ we have
$$
\frac{N_{\ell}^{\sigma^k}(\alpha)}{\sigma^{z{\de}}(N_{\ell}^{\sigma^k}(\alpha))} =  \prod_{j=1}^{z}\sigma^{z{\de}-j{\de}}\left(\frac{\sigma^{\ell k}(b)}{b}\right).
$$
Now let $s\in \mathbb{N}$ and define $z$ by $[sk]_{\de}=sk-z{\de}$.  Applying $\sigma^{[sk]_{\de}}=\sigma^{sk-z{\de}}$ to both sides yields
$$
\frac{\sigma^{[sk]_{\de}}(N_{\ell}^{\sigma^k}(\alpha))}{\sigma^{sk}(N_{\ell}^{\sigma^k}(\alpha))}
=
\prod_{j=1}^{z}\sigma^{sk-j{\de}}\left(\frac{\sigma^{\ell k}(b)}{b}\right).
$$
Finally, cross-multiplying and simplifying yields the desired equality
$$
\sigma^{[sk]_{\de}}(N_{\ell}^{\sigma^k}(\alpha))\prod_{j=1}^{z}\sigma^{sk-j{\de}}(b) = \sigma^{sk}(N_{\ell}^{\sigma^k}(\alpha))\prod_{j=1}^{z}\sigma^{sk-j{\de}+\ell k}(b).
$$
\end{proof}

We now demonstrate that the condition \eqref{E:conditiononassociativity} is exactly the key to the power associativity of a monomial in $\Petita{S}{b}$.
Recall that an element $\scz\in \Petita{S}{b}$ is said to be power-associative if $\scz^s \scz^t=\scz^t \scz^s$ for all $s,t\geq 1$, and this is equivalent to the statement that $\scz^{s+t}$ is well-defined.  Our first step is to evaluate the product $\scz^s$ as a left-nested product.

\begin{lemma}\label{L:firstpieceof3.3}
Let $\alpha, b\in \ring$ and ${\de}\in \N$.  Suppose neither $\alpha$ nor $b$ are zero-divisors and $1\leq k <{\de}$.  Let $\scz=\alpha t^k$, and let
$$
L(\scz,s) :=     \scz(\scz(\scz( \cdots (\scz))))
    $$
denote the left-nested product
of $\scz$ with itself $s$ times.  If
\begin{equation}\label{E:associativityrelation}
    \alpha b = \sigma^{\de}(\alpha)\sigma^k(b).
    \end{equation}
then we have
\begin{equation}\label{G(t)^s}
    L(\scz,s) = N_{s}^{\sigma^k}(\alpha)\prod_{i=1}^{\lfloor \frac{sk}{{\de}}\rfloor} \sigma^{sk-i{\de}}(b)t^{[sk]_{\de}}
    \end{equation}
    where $\lfloor \frac{sk}{\de}\rfloor=j$ when $j{\de}\leq sk<(j+1){\de}$ and $[sk]_{\de} = sk-\lfloor \frac{sk}{\de}\rfloor \de$ denotes the residue of $sk$ mod $\de$ in the interval $[0,\de-1]$.
\end{lemma}

\begin{proof} Suppose that \eqref{E:associativityrelation} holds.
    We prove the equality \eqref{G(t)^s} by induction on $s$.  The case $s=1$ is given.  Assume for some $s\geq 1$ that $L(\scz,s)$ is given by \eqref{G(t)^s}.  We wish to show that
    \begin{equation}\label{E:Lzs+1}
    L(\scz,s+1)= N_{s+1}^{\sigma^k}(\alpha)\prod_{i=1}^{\lfloor \frac{(s+1)k}{\de}\rfloor} \sigma^{(s+1)k-i\de}(b)t^{[(s+1)k]_{\de}}.
    \end{equation}
    We compute
    \begin{align*}
    L(\scz,s+1) &= \scz \cdot L(\scz,s)\\
    &=\alpha t^k \cdot N_{s}^{\sigma^k}(\alpha)\prod_{i=1}^{\lfloor \frac{sk}{\de}\rfloor} \sigma^{sk-i\de}(b)t^{[sk]_{\de}} \\
    &= \alpha \cdot \sigma^k\left( N_{s}^{\sigma^k}(\alpha)\right) \sigma^k\left(\prod_{i=1}^{\lfloor \frac{sk}{\de}\rfloor} \sigma^{sk-i\de}(b)\right)\cdot t^kt^{[sk]_{\de}}\\
    &= N_{s+1}^{\sigma^k}(\alpha) \prod_{i=1}^{\lfloor \frac{sk}{\de}\rfloor} \sigma^{(s+1)k-i\de}(b)\cdot t^kt^{[sk]_{\de}}.
    \end{align*}
    Note that since both $k$ and $[sk]_{\de}$ are less than $\de$, we simply have
    $$
    [(s+1)k]_{\de} = \begin{cases}
        k+[sk]_{\de} & \text{if $k+[sk]_{\de}<\de$, that is, if $\lfloor \frac{(s+1)k}{\de}\rfloor=\lfloor \frac{sk}{\de}\rfloor$};\\
        k+[sk]_{\de}-\de & \text{if $k+[sk]_{\de} > \de$, that is, if $\lfloor \frac{(s+1)k}{\de}\rfloor=\lfloor \frac{sk}{\de}\rfloor+1$.}
    \end{cases}
    $$
    In the former case we are done.  In the latter case, we apply the relation
    $$
    t^{k+[sk]_{\de}} = t^{k+[sk]_{\de}-\de}t^{\de}= t^{[(s+1)k]_{\de}}b=\sigma^{[(s+1)k]_{\de}}(b)t^{[(s+1)k]_{\de}},
    $$
    which gives the desired expression for $L(\scz,s+1),$  as required.
\end{proof}

Our next step is to show that under this hypothesis, the expression $\scz^s$ is indeed well-defined, that is, it is independent of the choice of placement of parentheses in its evaluation.

\begin{lemma}\label{L:secondpieceof3.3}
    In the setting of Lemma~\ref{L:firstpieceof3.3}, if \eqref{E:associativityrelation} holds, then $\scz$ is power-associative.  In  particular, for every $s>0$ we have $\scz^s=L(\scz,s)$, given by the formula \eqref{G(t)^s}.
\end{lemma}

\begin{proof}
        We assume that \eqref{E:associativityrelation} holds, and show that this implies power-associativity.  By induction it suffices to prove that
    for all $\ell,s\in \mathbb{N}$ we have
    $$
    L(\scz,s)L(\scz,\ell) = L(\scz,s+\ell).
    $$
This holds for $s=1$ and every $\ell\in \mathbb{N}$, by definition of the left-nested product.  We compute
\begin{align*}
L&(\scz,s)\; L(\scz,\ell)=\\
&=N_{s}^{\sigma^k}(\alpha)\prod_{i=1}^{\lfloor \frac{s k}{\de}\rfloor} \sigma^{s k-i\de}(b)t^{[s k]_{\de}}\cdot N_{\ell}^{\sigma^k}(\alpha)\prod_{i=1}^{\lfloor \frac{\ell k}{\de}\rfloor} \sigma^{\ell k-i\de}(b)t^{[\ell k]_{\de}}\\
&= N_{s}^{\sigma^k}(\alpha)\prod_{i=1}^{\lfloor \frac{s k}{\de}\rfloor} \sigma^{s k-i\de}(b)\cdot
\sigma^{[s k]_{\de}}\left(N_{\ell}^{\sigma^k}(\alpha)\prod_{i=1}^{\lfloor \frac{\ell k}{\de}\rfloor} \sigma^{\ell k-i\de}(b)\right) \cdot t^{[s k]_{\de} + [\ell k]_{\de}}
\end{align*}
By Lemma~\ref{L:keystep}, our hypothesis \eqref{E:associativityrelation} implies
$$
\sigma^{[sk]_\de}(N_{\ell}^{\sigma^k}(\alpha))\prod_{j=1}^{\lfloor \frac{s k}{\de}\rfloor}\sigma^{sk-j\de}(b) = \sigma^{sk}(N_{\ell}^{\sigma^k}(\alpha))\prod_{j=1}^{\lfloor \frac{s k}{\de}\rfloor}\sigma^{sk-j\de+\ell k}(b).
$$
Thus
\begin{align*}
  L&(\scz,s)\; L(\scz,\ell)=\\
  &= N_{s}^{\sigma^k}(\alpha) \left(\prod_{i=1}^{\lfloor \frac{s k}{\de}\rfloor} \sigma^{s k-i\de}(b)\cdot
\sigma^{[s k]_{\de}}(N_{\ell}^{\sigma^k}(\alpha)) \right) \sigma^{[sk]_{\de}}\left(\prod_{i=1}^{\lfloor \frac{\ell k}{\de}\rfloor} \sigma^{\ell k-i\de}(b)\right) \cdot t^{[s k]_{\de} + [\ell k]_{\de}}\\
&= N_{s}^{\sigma^k}(\alpha)
\left( \sigma^{sk}(N_{\ell}^{\sigma^k}(\alpha))\prod_{j=1}^{\lfloor \frac{s k}{\de}\rfloor}\sigma^{sk-j\de+\ell k}(b)
\right)
\prod_{i=1}^{\lfloor \frac{\ell k}{\de}\rfloor} \sigma^{[sk]_\de+\ell k-i\de}(b) \cdot t^{[s k]_{\de} + [\ell k]_{\de}}.
\end{align*}
Since $[sk]_{\de}=sk-\lfloor \frac{sk}{\de}\rfloor \de$, we may re-index the last product to yield
\begin{align*}
&= N_{s+\ell}^{\sigma^k}(\alpha)
\prod_{j=1}^{\lfloor \frac{s k}{\de}\rfloor}\sigma^{sk-j\de +\ell k}(b)
\prod_{j=\lfloor \frac{s k}{\de}\rfloor+1}^{\lfloor \frac{s k}{\de}\rfloor+\lfloor \frac{\ell k}{\de}\rfloor}\sigma^{sk-j\de+\ell k}(b)  \cdot t^{[s k]_{\de} + [\ell k]_{\de}}\\
&= N_{s+\ell}^{\sigma^k}(\alpha)\prod_{j=1}^{\lfloor \frac{s k}{\de}\rfloor+\lfloor \frac{\ell k}{\de}\rfloor} \sigma^{(s+\ell)k-j\de }(b)  \cdot t^{[s k]_{\de} + [\ell k]_{\de}}.
\end{align*}
If $[sk]_{\de}+[\ell k]_{\de} = [(s+\ell)k]_{\de}$ then this is simply
$$
N_{s+\ell}^{\sigma^k}(\alpha)\prod_{j=1}^{\lfloor \frac{(s+\ell)k}{\de}\rfloor } \sigma^{(s+\ell)k-j\de}(b)  \cdot t^{[(s+\ell) k]_{\de}} = L(\scz,s+\ell).
$$
Otherwise, we must have
$$
\lfloor \frac{(s+\ell)k}{\de}\rfloor  = \lfloor \frac{s k}{\de}\rfloor+\lfloor \frac{\ell k}{\de}\rfloor+1,
$$
whence $
[sk]_{\de}+[\ell k]_{\de} = [(s+\ell)k]_{\de} + \de.
$
Since $[(s+\ell) k]_{\de}=(s+\ell)k-\lfloor \frac{(s+\ell)k}{\de}\rfloor \de$ we have
$$
t^{[s k]_{\de} + [\ell k]_{\de}}=t^{[(s+\ell) k]_{\de}+\de}= \sigma^{(s+\ell)k-\lfloor \frac{(s+\ell)k}{\de}\rfloor \de}(b)t^{[(s+\ell) k]_{\de}},
$$
which again yields $L(\scz,s)L(\scz,\ell)=L(\scz,s+\ell).$  Thus the desired equality holds for all $s,\ell \in \mathbb{N}$, and $\scz$ is power associative.
\end{proof}

Having shown that the key condition \eqref{E:associativityrelation} implies the power associativity of $\scz=\alpha t^k$, we finally prove our main theorem of the section.

\begin{theorem}\label{T:powerassociative}
    Let $\alpha, b\in \ring$ and ${\de}\in \N$.  Suppose neither $\alpha$ nor $b$ are zero-divisors and $1\leq k <{\de}$.
    The monomial $\alpha t^k \in \Petita{S}{b}$
    is power-associative if and only if
    \begin{equation}\label{E:G(t)again}
    \alpha b = \sigma^{\de}(\alpha)\sigma^k(b).
    \end{equation}
\end{theorem}

\begin{proof}
By Lemma~\ref{L:secondpieceof3.3}, the condition \eqref{E:associativityrelation} is sufficient; we now prove it is necessary.  Suppose $\scz$ is power-associative.  Then  $\scz\scz^r=\scz^r\scz$, for $r>0$ the least integer such that $rk\geq \de$.  We show that then  \eqref{E:associativityrelation} holds.  We begin by noting that $\scz^r$ is well-defined, since for all $s,\ell\geq 1$ such that $s+\ell=r$,
$$
\scz^s=N_s^{\sigma^k}(\alpha)t^{sk}, \quad \text{and} \quad \scz^\ell = N_\ell^{\sigma^k}(\alpha)t^{\ell k},
$$
by the associativity of powers less than $\de$ in $\Petita{S}{b}$.  Thus we readily compute as above that
$$
\scz^r=\scz^s\scz^\ell= N_s^{\sigma^k}(\alpha)\sigma^{sk}(N_\ell^{\sigma^k}(\alpha))t^{(s+\ell)k}=N_r^{\sigma^k}(\alpha)\sigma^{rk-\de}(b)t^{rk-\de},
$$
independently of the choice of association, and this is equal to $L(\scz,r)$.

Now by our choice of $r$ (and $k$), we have $\de<(r+1)k < 3\de$.

\textbf{Case 1:}  If $\de<(r+1)k<2\de$, then  $\lfloor (r+1)k/\de \rfloor=1$ and $[(r+1)k]_{\de}=(r+1)k-\de$.  Thus we compute, using  \eqref{E:Lzs+1} with $s=r$, that
$$
\scz \scz^r=L(\scz,r+1)=N_{r+1}^{\sigma^k}(\alpha)\sigma^{(r+1)k-\de}(b)t^{(r+1)k-\de}
$$
whereas
\begin{equation}\label{E:zrzproof}
 \scz^r \scz=N_r^{\sigma^k}(\alpha) \sigma^{rk-\de}(b)t^{rk-\de} \alpha t^k
= N_r^{\sigma^k}(\alpha) \sigma^{rk-\de}(\alpha)\sigma^{rk-\de}(b)t^{(r+1)k-\de}.
\end{equation}
As these are assumed equal, we may equate the coefficients of $t^{(r+1)k-\de}$.  Since $\alpha$ (and hence  $N_r^{\sigma^k}(\alpha)$) is not a zero-divisor, we may cancel it to obtain the equality
\begin{equation}\label{E:intermediate}
\sigma^{rk}(\alpha)\sigma^{rk-\de}(\sigma^k(b))=\sigma^{rk-\de}(\alpha)\sigma^{rk-\de}(b).
\end{equation}
Applying the automorphism $\sigma^{\de-rk}$ to both sides gives \eqref{E:associativityrelation}, as required.

\textbf{Case 2:}
When instead $2\de \leq (r+1)k <3\de$, then $\lfloor (r+1)k/\de \rfloor=2$ and $[(r+1)k]_{\de}=(r+1)k-2\de$.  Thus \eqref{E:Lzs+1} yields instead
$$
\scz \scz^r=L(\scz,r+1)=N_{r+1}^{\sigma^k}(\alpha)\sigma^{(r+1)k-\de}(b)\sigma^{(r+1)k-2\de}(b)t^{(r+1)k-2\de}.
$$
The term $\scz^r \scz$ was computed in \eqref{E:zrzproof},  but since $(r+1)k-\de\geq \de$, we must use the relation $t^{(r+1)k-\de}=t^{(r+1)k-2\de}b=\sigma^{(r+1)k-2\de}(b)t^{(r+1)k-2\de}$ in $\Petita{S}{b}$ to simplify it to the final form, which gives
$$  \scz^r \scz=
N_r^{\sigma^k}(\alpha) \sigma^{rk-\de}(\alpha)\sigma^{rk-\de}(b)\sigma^{(r+1)k-2\de}(b) t^{(r+1)k-2\de}.
$$
Equating coefficients and cancelling the non-zero-divisor $N_r^{\sigma^k}(\alpha)\sigma^{(r+1)k-2\de}(b)$  from both sides yields  \eqref{E:intermediate}, from which we again deduce \eqref{E:associativityrelation}, as required.
\end{proof}

The theorem has immediate consequences for the existence of monomial homomorphisms.

\begin{corollary}\label{C:StsigmatoSb}
    Let $k\in \N$ and $\alpha \in \ring$ not be a zero divisor. Let $\tau$ be an automorphism of $\ring$ commuting with the order-$\os$ automorphism $\sigma$.  Then the map
    $$
    G =G_{\tau,\alpha,k}: \Stsigma \to \Petita{S}{b}
    $$
    given by $G(a)=\tau(a)$ and $G(t) = \alpha t^k$ induces a ring homomorphism leaving $\fixedts$ fixed 
    if and only if
    $$
    b \in \ring_0^\times, \quad k\equiv 1 \mod \os, \quad \text{and}\quad \os|\de.
    $$
\end{corollary}

\begin{proof}
    By Lemma~\ref{L:Ktsigmatoitself}, the condition $k\equiv 1 \mod \os$ ensures that $G(ta)=G(\sigma(a)t)$.    If $G$ is well-defined, then $G(t)$ must be power-associative.  By Lemma~\ref{L:firstpieceof3.3}, it follows that $G(t^s)=\lambda t^{[sk]_{\de}}$ for some $\lambda \in \ring^\times$.
    Thus the requirement that
    $G(t^s)G(a)=G(\sigma^s(a))G(t^s)$ for all $a\in \ring$ and all $s\in \mathbb{N}$ is equivalent to
    $$
    \sigma^{[sk]_{\de}}(\tau(a)) = \tau(\sigma^{s}(a)).
    $$
    Since $\tau$ commutes with $\sigma$ this holds if and only if  $s-[sk]_{\de}\equiv 0 \mod \os$; since $k\equiv 1 \mod \os$ this implies $\os\mid\de$.  Thus $\de= e \os$ for some positive integer $e$.
    Under this hypothesis, the condition $\alpha b = \sigma^{\de}(\alpha)\sigma^k(b)$
     simplifies to  $\alpha b = \sigma^{e\os}(\alpha)\sigma^k(b)=\alpha \sigma^k(b)$,
     since $\sigma$ has order $\os$.  Since $\alpha$ is not a zero-divisor, this implies that $\sigma(b)=b$ or $b\in \ring_0^\times$, and
    thus $\Petita{S}{b}$ is associative.
    Consequently
    the function $G=G_{\tau,\alpha,k}$
    descends to a well-defined ring homomorphism $G:\Stsigma\to \Petita{S}{b}$.
\end{proof}

Note that the homomorphism $G$ in Corollary~\ref{C:StsigmatoSb} is a surjection
when additionally $\alpha \in \ring^\times$ and $\gcd(k,\de)=1$.

The conditions in Corollary~\ref{C:StsigmatoSb} are quite restrictive, reflecting the constrained nature of these nonassociative algebras. For example, one deduces that even with $k=1$, the map $G_{\tau,\alpha,1}$ is a homomorphism from $\Stsigma\to \Petita{S}{b}$ only when $b\in \ring_0^\times$ and $\os|\de$.  This is in fact to be expected: $\Stsigma$ is associative whereas $\Petita{S}{b}$
is associative if and only if these two strong conditions hold.

\section{Monomial homomorphisms between Petit algebras} \label{S:4}

We now turn to the question of when there exist monomial homomorphisms $G_{\tau,\alpha,k}$ from a Petit algebra $\Petit{S}{f}$,
where $f$ has degree $\de$, into the algebra $\Petita{S}{b}$.

 It is not evident that such a homomorphism $G$ is forced to be the reduction of a homomorphism $\tilde{G}:\Stsigma \to \Petita{S}{b}$ as in Corollary~\ref{C:StsigmatoSb}.
Namely, in  $\Petit{S}{f}$,
the associativity of all powers of $t$ is a strong condition \cite[p.~13-06]{Petit1967}: it is equivalent to $t^{\de} t=t t^{\de}$ and to $f(t) t \in \Stsigma f(t)$.  In particular, if all powers of $t$ are associative in  $\Petit{S}{f}$,
then they are contained in the right nucleus of $\Petit{S}{f}$,
which is an associative subalgebra, and thus it is only in this case that Theorem \ref{T:powerassociative} applies directly, as in the following example.

\begin{example}
    Suppose $f(t)\in \ring_0[t]$.  Then the associative algebra $\ring_0[t]/\langle f(t)\rangle$ is a subalgebra of $\Petit{S}{f}$
    containing all powers of $t$, and therefore if $G:\Petit{S}{f}\to \Petita{S}{b}$ is a  monomial homomorphism,
    then $G(t)$ must be fully power associative.  It thus follows from  Theorem~\ref{T:powerassociative} that we require
    $$
\alpha b = \sigma^{\de}(\alpha) \sigma^k(b)
$$
and by \eqref{G(t)^s}, for  every $s\in \mathbb{N}$, the expression $G(t^s)=G(t)^s$ is equal to
\begin{equation}\label{G(t)^sextra}
    N_{s}^{\sigma^k}(\alpha)\prod_{i=1}^{\lfloor \frac{sk}{\de}\rfloor} \sigma^{sk-i\de}(b)t^{[sk]_{\de}}.
\end{equation}
\hfill \qed
\end{example}

In contrast, when $\Petit{S}{f}$
is proper nonassociative, then  $t^s$ is only well-defined for $0\leq s <\de$, and the full power associativity of Theorem~\ref{T:powerassociative} is not required.

To illustrate this, we derive the following generalization of Theorem~\ref{T:general_isomorphism_theorem2}, about the existence of monomial homomorphisms of degree $1$, between Petit algebras of our chosen form.

\begin{proposition}\label{L:Gtaualpha}
    Let $\tau\in \Aut(\ring)$ commute with $\sigma$ and suppose $\alpha\in \ring$ is not a zero divisor.  Then the map
    $$
    G_{\tau,\alpha}: \Petita{S}{a} \to \Petita{S}{b}
    $$
    induces a homomorphism  if and only if
    $$
    \tau(a) = N_{\de}^\sigma(\alpha) b,
    $$
    and this is an isomorphism if and only if $\alpha \in \ring^\times$.
\end{proposition}

In particular, there is no requirement on the relationship between $\de$ and $\os$.

\begin{proof}
    Write $G=G_{\tau,\alpha}$; then $G(t)=\alpha t$ and so for all $1\leq s < \de$ we unambiguously have $G(t^s)=N_s^\sigma(\alpha)t^s$ since these powers are associative in both rings.   Thus since $G$ acts by $\tau$ on $\ring$, $G$ extends  to an $\ring_0^\tau$-linear isomorphism of the underlying $\ring_0^\tau$-modules.

    Furthermore, for all $c\in \ring$ and $0<s<\de$ we have $t^sc=\sigma^s(c)t^s$.  We find
    $$
    G(t^s)G(c)=N_s^\sigma(\alpha)t^s \tau(c)=N_s^\sigma(\alpha)\sigma^s(\tau(c))t^s = \tau(\sigma^s(c))N_s^\sigma(\alpha)t^s=G(\sigma^s(c))G(t^s)
    $$
    and thus this identity always holds.

    We now verify that $G(pq)=G(p)G(q)$ for all $p,q\in \Petita{S}{a}$. By the above, it suffices to choose arbitrary $1\leq s,\ell <\de$ and prove this for $p=t^s$, $q=t^\ell$.

    Consider the product $t^{s}t^\ell$ in $\Petita{S}{a}$.  If $s+\ell<\de$ the result is $t^{s+\ell}$ and we correspondingly have
    $$
    G(t^s)G(t^\ell) = N_s^\sigma(\alpha)t^s N_\ell^\sigma(\alpha)t^\ell = N_{s}^\sigma(\alpha)\sigma^s(N_\ell^\sigma(\alpha))t^{s+\ell} = N_{s+\ell}^\sigma(\alpha)t^{s+\ell}
    $$
    which is equal to $G(t^{s+\ell})$.
    If $s+\ell\geq \de$, however, then we have $t^st^\ell = t^{s+\ell-\de}a=\sigma^{s+\ell-\de}(a)t^{s+\ell-\de}$ whose image under $G$ is
    $$
    G(t^st^\ell)=\tau(\sigma^{s+\ell-\de}(a))N_{s+\ell-\de}^\sigma(\alpha)t^{s+\ell-\de}
    $$
    while
    $$
    G(t^s)G(t^\ell) = N_{s+\ell}^\sigma(\alpha)t^{s+\ell-\de}b=N_{s+\ell}^\sigma(\alpha)\sigma^{s+\ell-\de}(b)t^{s+\ell-\de}.
    $$
    Since $N_{s+\ell}^\sigma(\alpha)=N_{s+\ell-\de}^\sigma(\alpha)\sigma^{s+\ell-\de}(N_{\de}^\sigma(\alpha))$ by \eqref{E:Normrelation}
   and  $N^\sigma_{s+\ell-\de}(\alpha)$ is not a zero-divisor, we deduce that these two expressions are equal whenever
    $$
    \tau(\sigma^{s+\ell-\de}(a))=\sigma^{s+\ell-\de}(N_{\de}^\sigma(\alpha))\sigma^{s+\ell-\de}(b).
    $$
    Therefore, since $\tau$ and $\sigma$ commute, our map is a homomorphism of $\ring_0^\tau$-algebras if and only if $\tau(a) = N_{\de}^\sigma(\alpha) b$.  It is invertible exactly when $\alpha^{-1}\in \ring$.
\end{proof}

For example, when $\os=\de$, $\ring=\field$ is a Galois field extension of $\field_0$ and $\tau\in {\rm Gal}(\field/\field_0)$, we recover the condition that $\tau(a) \in b N_{\field/\field_0}(\field^\times)$ from  \cite[Corollary 32]{BrownPumpluen2018}.

We now turn to the case that the degree $k$ of the monomial homomorphism $G_{\tau,\alpha,k}$ is strictly greater than $1$.
We begin with a lemma that is a version of Lemma~\ref{L:Ktsigmatoitself} for the present setting.

\begin{lemma}\label{L:k=1}
    Suppose $f,h\in \Stsigma$ are monic polynomials of degree $\de$.  If for some non-zero-divisor $\alpha\in \ring$ and degree $k>1$, the map $G_{\tau,\alpha,k}:\Petit{S}{f}\to\Petit{S}{h}$
    is a monomial homomorphism, then $k\equiv 1 \mod \os$; if it is an isomorphism then $\gcd(k,\de)=1$ and $\alpha\in \ring^\times$.
\end{lemma}

\begin{proof}
    If $G=G_{\tau,\alpha,k}$ is a homomorphism, then the identity $G(td)=G(\sigma(d)t)$ implies $\alpha t^k \tau(d) = \tau(\sigma(d)) \alpha t^k$ for all $d\in \ring$, whence we must have
$$
\sigma^k(\tau(d)) = \tau(\sigma(d))
$$
for all $d\in \ring$.  Thus since $\tau$ and $\sigma$ commute, $G(td)=G(\sigma(d)t)$ is equivalent to  $k \equiv 1 \mod \os$.

If $G$ is to be an isomorphism, there must exist an inverse homomorphism $H$.
Evidently $H(d)=\tau^{-1}(d)$ for all $d\in \ring$.  Since $G$ is monomial, so is its inverse, and thus we may assume $H(t)=\beta t^\ell$ for some $\beta \in \ring^\times$ and $1\leq \ell < \de$.  Since $H(G(t))=t$ we must have $\ell k \equiv 1 \mod \de$ and in particular $\gcd(k,\de)=1$. Moreover, if $\alpha$ were not invertible, then $t^k$ could not lie in the image of $G$.
\end{proof}

\begin{theorem}\label{prop:G}
    Suppose $\tau\in \Aut(\ring)$ commutes with $\sigma$ and $\alpha \in \ring$ is not a zero-divisor.  Then for any $2\leq k < \de$, the map
    $$
    G=G_{\tau,\alpha,k}:\Petita{S}{a} \to \Petita{S}{b}
    $$
    is a homomorphism of rings if and only if all of the following conditions hold:
    \begin{enumerate} \setlength{\itemsep}{0pt}
        \item \label{C:k=1} $k\equiv 1 \mod \os$;
        \item \label{C:n|m} $\os\mid \de$;
        \item \label{C:inF} $a, b\in \ring_0$;
        \item \label{C:norm} $(N_{\ring/\ring_0}(\alpha))^{\de/\os}b^k=\tau(a)$.
    \end{enumerate}
    It is an isomorphism if and only if, in addition, we have
    \begin{enumerate} \setcounter{enumi}{4}
        \item \label{C:relprime} $\gcd(k,\de)=1$ and $\alpha\in \ring^\times$.

    \end{enumerate}
\end{theorem}

\begin{proof}
Let $\tau\in \Aut(\ring)$ commute with $\sigma$. Fix $\alpha\in \ring^\times$ and  $2\leq k<\de$.  Suppose first that the map $ G=G_{\tau,\alpha,k}:\Petita{S}{a} \to \Petita{S}{b}$
defined by the rule
\begin{equation}\label{E:formulaG(t)}
G\left( \sum_{i=0}^{\de-1}a_it^i \right) = \sum_{i=0}^{\de-1}\tau(a_i) (\alpha t^k)^i
\end{equation}
is a well-defined homomorphism.

Conditions (\ref{C:k=1}) and (\ref{C:relprime}) follow from Lemma~\ref{L:k=1}.

For $G$ to be well-defined, we require for all $0<s <\de$ that $G(t^s)=G(t)^s$ is well-defined.  Since $k>1$, the minimal value of $r$ such that $rk\geq \de$ satisfies $0<r<\de$ and thus by Theorem~\ref{T:powerassociative}, we conclude that $\alpha b = \sigma^{\de}(\alpha)\sigma^k(b)$.  Conversely, Theorem~\ref{T:powerassociative} further assures us that this hypothesis implies all powers of $\alpha t^k$ are well-defined so \eqref{E:formulaG(t)} gives a well-defined map.

Next, we note from Theorem~\ref{T:powerassociative} that for any $0<s<\de$, the element $G(t^s)=G(t)^s \in \Petita{S}{b}$ will be a monomial of degree $[sk]_{\de}$, the residue of $sk$ mod $\de$. In order for $G$ to further preserve the relation $t^sd=\sigma^s(d)t^s$ for all $d\in \ring$ we require $G(t^s)\tau(d) = \tau(\sigma^s(d))G(t^s)$, which will hold if and only if
$$
\sigma^{[sk]_{\de}}(\tau(d))=\tau(\sigma^{s}(d))
$$
for all $d\in \ring$. Since $k\equiv 1 \mod \os$, and $\tau\sigma=\sigma\tau$, this is equivalent to the requirement that $\sigma^{sk}=\sigma^{[sk]_{\de}}$ for all $0<s<\de$, which holds if and only if $\sigma^{\de}=1$.  Thus a second necessary condition is (\ref{C:n|m}), that $\os|\de$, and it in return implies $G(t^s)\tau(d) = \tau(\sigma^s(d))G(t^s)$ for all $0<s<\de$ and $d\in \ring$.

Since $\os\mid\de$ and $k\equiv 1 \mod \os$, the relation $\alpha b = \sigma^{\de}(\alpha)\sigma^k(b)$ simplifies to $\sigma(b)=b$, or $b\in \ring_0$, yielding part of the condition (\ref{C:inF}).

Next, $G$ must satisfy $G(t^{\de}-a)=0$.  By the above, we may evaluate $G(t^{\de})=G(t)^{\de}$ using the formula \eqref{G(t)^s}.  With $\sigma^k=\sigma$ and $\sigma^{\de}=1$, this yields
$$
G(t)^{\de} = N_{\de}^{\sigma}(\alpha)\left(\prod_{i=1}^{k} b\right)t^{0} = (N^\sigma_{\os}(\alpha))^{\de/\os} b^k = (N_{\ring/\ring_0}(\alpha))^{\de/\os}b^k.
$$
Now since $G(t^{\de})=G(a)$, this is equivalent to the condition (\ref{C:norm})
$$
(N_{\ring/\ring_0}(\alpha))^{\de/\os}b^k=\tau(a).
$$
In particular $\tau(a)\in \ring_0$, so $a\in \ring_0$, which is the rest of condition (\ref{C:inF}).

Note that if $G$ is an isomorphism, its inverse map must be a homomorphism.  If $G^{-1}(t)=\beta t^\ell$, then $\ell \equiv 1 \mod \os$; in fact, $\ell$ is simply the inverse of $k$ mod $\de$ (as guaranteed by condition (\ref{C:relprime})).

Finally, if $k,\de,\os,\alpha,a,b,\tau$ satisfy the first four conditions, then the $\ring_0^\tau$-linear map $G:\Petita{S}{a}\to \Petita{S}{b}$
given by \eqref{E:formulaG(t)} is a well-defined homomorphism of $\ring_0^\tau$-modules, and respects the relations $G(t^{\de})=G(a)$ and $G(t^su)=G(\sigma^s(u)t^s)$ for all $0<s<\de$ and all $u\in \ring$.  Therefore in particular it satisfies $G(p)G(q)=G(pq)$ and $G(p+q)=G(p)+G(q)$ for all $p,q\in \Petita{S}{a}$, which implies it is a homomorphism of $\ring_0^\tau$-algebras.  The final condition ensures it is an isomorphism.
\end{proof}

Note that when $\os\geq \de -1$, Theorem \ref{prop:G} shows the only isomorphisms between $\Petita{S}{a}$ and $\Petita{S}{b}$
are monomial of degree one, recovering the result in Theorem \ref{T:BPonlyonesareGtaualpha}.

An immediate consequence is the realization that there are very few monomial homomorphisms of degree greater than one in this setting.

\begin{corollary}\label{C:onlyweightone}
    Suppose $\os\nmid \de$, or that one of $a$ or $b$ is not in $\ring_0$.  Then  any Hamming-weight preserving homomorphism
    $$
    G : \Petita{S}{a}\to \Petita{S}{b}
    $$
    will be of degree one, that is, will be of the form $G=G_{\tau,\alpha}$ for some $\tau\in \Aut(S)$ commuting with $\sigma$, and some $\alpha \in \ring$ that is not a zero divisor.
\end{corollary}

\begin{proof}
    If $G$ is Hamming-weight preserving, then it must send $t$ to an element of Hamming-weight $1$, hence it is monomial.  Further, the coefficient $\alpha$ in this case cannot be a zero divisor, or else there would be some $c\in \ring$ for which $G(ct)=0$.      By Theorem~\ref{prop:G}, the degree of $G$ cannot be greater than one.
\end{proof}

When $a,b\in \ring_0$, on the other hand, Theorem~\ref{prop:G} gives a plethora of monomial homomorphisms $G:\Petita{S}{a}\to \Petita{S}{b}$
of degree $k>1$.

\begin{example}\label{Examplewehave}
    Let us first illustrate Theorem~\ref{prop:G} in a general example that holds over every ring $\ring$ (including particularly any finite field or  finite chain ring, for example).  Suppose $\os=2,\de=4$ (so that $\os\mid \de$).  We require a power $2\leq k <\de$ such that $k\equiv 1 \mod \de$ and $\gcd(k,\de)=1$, so we may choose $k=3$. Choose $b\in \ring_0^\times$ and $\alpha\in \ring^\times$.  Then if we set $r=\alpha \sigma(\alpha)$ and $\tau=\id$, then the key relation
    $$
    \tau(a) = N_{S/S_0}(\alpha))^{\de/\os} b^k
    $$
    defines a unique element
    $$
    a = r^2b^3.
    $$
    Since $r,b\in S_0^\times$, we have $a\in S_0^\times$. Therefore by Theorem~\ref{prop:G}, the map $G=G_{\id,\alpha,3}$ induces an isomorphism
    $$
  \Petit{S}{t^4 -r^2b^3}\to \Petit{S}{t^4 -b}
    $$
    which is monomial of degree strictly greater than $1$.  More precisely, it is the map
    $$
    \sum_{i=0}^3 a_it^i \mapsto a_0 + a_3\alpha^2\sigma(\alpha)b^2t+a_2\alpha\sigma(\alpha)bt^2 + a_1\alpha t^3.
    $$
\end{example}

The next example proves that monomial isomorphisms of degree $k>1$ can capture strictly more isometries than monomial isomorphisms of degree $1$.

\begin{example}
Consider the finite chain ring $S=\Z_9[i]$ where $\sigma(i)=-i$ so $S_0=\Z_9$.  Then $\de=6$, $\os=2$ and $k=5$ meet the requirements of Theorem~\ref{prop:G}.  Choosing $b=2\in S_0$ and $\alpha=1$,  we may set
$$
a = N_{S/S_0}(1)^3b^5=2^5=5.
$$
Therefore, by Theorem~\ref{prop:G},  $G_{\id,1,5}:\Petita{S}{5}\to\Petita{S}{2}$ is an isometry, that is, a Hamming-weight preservering isomorphism.  However, since $N_{S/S_0}(S^\times)^3 = \{\pm 1\}$, and $2\neq \pm 5$, there is \emph{no} monomial isomorphism of degree $1$ between these Petit algebras. Thus, the degree $1$ maps do not capture all isometries.
\end{example}

\section{Applications}\label{S:applications}

In Section~\ref{S:4}, we classified all monomial homomorphisms between Petit algebras of the form $\Petita{S}{a}$.
Since the principal left ideals of these algebras give rise to skew constacyclic codes, we have classified the Hamming-weight preserving isomorphisms, \emph{i.e.} isometries, between their corresponding skew constacyclic codes.

\subsection{Isometric constacyclic codes}\label{subsec:iso}

In \cite{OuazzouNajmeddineAydin2025}, the authors explore Hamming-weight preserving isomorphisms between nonassociative algebras of the form $$\Petita{K}{b}$$ over a finite field $\field=\mathbb{F}_{p^r}$ with respect to a Frobenius automorphism $\sigma(a)=a^{p^s}$, 
as an important step towards classifying skew constacyclic codes of length $\de$ with identical performance parameters.  This is necessary to avoid duplicating already existing codes. To this end, they proposed notions of $(\de,\sigma)$-equivalence and $(\de,\sigma)$-isometry, respectively.

We say two rings $\Petita{K}{a}$ and $\Petita{K}{b}$ are \emph{$(\de,\sigma)$-equivalent} if they are isomorphic via an isomorphism of the form $G_{\id,\alpha}$ for some $\alpha\in \field^\times$, and
\emph{$(\de,\sigma)$-isometric} if they are isomorphic via an isomorphism of the form $G_{\id,\alpha,k}$ for some $\alpha\in \field^\times$ and $k\geq 1$.

By Corollary~\ref{C:onlyweightone}, we conclude that these notions coincide whenever $a$ or $b$ lies in $\field_0=\F_{p^s}$, and whenever the order of $\sigma$ 
does not divide $\de$.

Inspired by their work, we now propose a more complete notion of equivalence for such rings.

\begin{definition}\label{D:isometric}
    Two rings $\Petita{S}{a}$ and $\Petita{S}{b}$
    are called \emph{isometric} if there exists a monomial isomorphism $G_{\tau,\alpha,k}$ between them for some $\tau\in \Aut(\ring)$ commuting with $\sigma$, $\alpha \in \ring^\times$ and $1\leq k<\de$. As a special case, we say they are \emph{equivalent} if they are isomorphic via an isomorphism of the form $G_{\tau,\alpha}$ for some $\alpha\in \ring^\times$ and some $\tau\in \Aut(\ring)$ commuting with $\sigma$.

Two classes ${\bf C}_a$  and ${\bf C}_b$ of skew $(\sigma,a)$-constacyclic codes of length $\de$, respectively skew $(\sigma,b)$-constacyclic codes of length $\de$, are called
 \emph{isometric} if there exists an isometry
$$G_{\tau,\alpha,k}:  \Petita{S}{a}\to\Petita{S}{b}
$$ and \emph{equivalent}  if there exists an isometry
$$G_{\tau,\alpha}: \Petita{S}{a}\to\Petita{S}{b}.
$$
\end{definition}

Thus if two rings are $(\de,\sigma)$-isometric, or $(\de,\sigma)$-equivalent, then they are isometric, and if they are  $(\de,\sigma)$-equivalent then they are equivalent. Moreover, this definition captures all of the possible Hamming-weight preserving isomorphisms between these ambient rings, and as the isometries $G_{\tau,\alpha,k}: \Petita{S}{a}\to\Petita{S}{b}$
preserve the dimension of the codes as well, this makes it the correct notion of equivalence between codes. This seems counterintuitive if $k>1$ so let us give a clarifying example.

\begin{example}\label{Egwb}
    We now specialize Example~\ref{Examplewehave} to the case that $S$ is finite (for example, a finite field) so that the ideals of $\Petita{S}{b}$ are skew $(4,\sigma)$-constacyclic codes over $S$.   Choose $w,b\in S_0^\times$ to satisfy $b^3=w^4$.  Then choosing $\alpha=1$ yields, by Example~\ref{Examplewehave}, a monomial isomorphism $G=G_{\id,1,3}:\Petit{S}{t^4-b^3} \to \Petit{S}{t^4-b}$ of degree $3$.  In this case, since
    $$t^4-b^3=(t^3+wt^2+w^2t+w^3)(t-w),$$
    we see that $t-w$ is a right divisor of $t^4-b^3$ in  $\Stsigma$, so it generates a left ideal $C$ (that is, a skew $(4,\sigma)$-constacyclic code) of dimension $\de - \deg(t-w)=4-1=3$ (\emph{e.g.}, see \cite[Theorem 2.9 (i)]{Pumpluen2025}).

    On the other hand,  $G(t-w)=t^3-w$ is not a divisor of $t^4-b$ in $\Stsigma$ but nevertheless (as one may verify directly)
    $$
    G(t^3+wt^2+w^2t+w^3)G(t-w)=0 \quad \text{in} \quad \Petit{S}{t^4-b}.
    $$
    It would be erroneous to conclude that $\dim(G(C))=\de-\deg(G(t-w))$, since although $G(t-w)$ generates the ideal $G(C)$, it will not be a polynomial of least degree in this ideal.   Instead, we identify a polynomial $c(t)$ of least degree in $G(C)$ (by computing $t^s(t^3-w)\in G(C)$ for $s=0,1,2,\cdots$).  By these means, one finds that $G(C)$ is also generated by its skew-cyclic shift
    $$
    (-w^{-1}t)(t-w) = t-bw^{-1} \in \Petit{S}{t^4-b},
    $$
    and in particular is a code with the same parameters as $C$, as expected.  We verify explicitly that this new choice of generator is indeed a right divisor of $t^4-b$ in $\Stsigma$ since
    $$
    t^4-b = (t^3-bw^{-1}t^2-b^2w^{-2}t+b^3w^{-3})(t-bw^{-1}).
    $$
    \hfill \qed
\end{example}

Now suppose two classes ${\bf C}_a$ and ${\bf C}_b$ are  isometric (of the same length $\de$).
Then the skew constacyclic codes in ${\bf C}_a$ and ${\bf C}_b$ are $S$-submodules of $S^{\de}=\Stsigma_{\de}$.
 Every skew $(\sigma,a)$-constacyclic code $C$ in ${\bf C}_a$ is generated by a monic $g\in \Stsigma$ that is a right divisor of $t^\de-a$ in $\Stsigma$, therefore every skew $(\sigma,a)$-polycyclic code $C$ in ${\bf C}_b$ corresponds to a free (and cyclic) $S$-submodule $\Stsigma g/\Stsigma(t^\de-a)$. We also know that $C$ has dimension $k=\de-\deg(g)$.

\begin{lemma}\label{le:pres}
   If two classes ${\bf C}_a$ and ${\bf C}_b$ are  isometric via some monomial isomorphism $G_{\tau,\alpha,\ell}$ of degree $\ell$, then the skew $(\sigma,a)$-constacyclic  codes in ${\bf C}_a$ have the same Hamming distance, dimension and length as
  the  skew $(\sigma,b)$-constacyclic  codes in ${\bf C}_b$.
   \end{lemma}

 \begin{proof}
By assumption, there exists an isomorphism of nonassociative ambient Petit rings $G_{\tau,\alpha,\ell}:\Petita{S}{a} \to \Petita{S}{b} $. This immediately implies that the length of the codes in both classes is $\mathbf{n}$, and so the ring isomorphism $G_{\tau,\alpha,\ell}$ preserves the length of codes. Since the morphism is monomial, it sends monomials of $\Petita{S}{a}$ to monomials of $\Petita{S}{b}$ and so preserves Hamming distance.

Let $g$ be a right divisor of $t^\de-a$ in $\Stsigma$; then the ideal generated by $g$ in  $\Petita{S}{a}$,
call it $C$,  is a skew $(\sigma,a)$-constacyclic code and all $(\sigma,a)$-constacyclic codes arise in this way \cite[Section 2.2]{Pumpluen2025}.  The image of a right ideal of $\Petita{S}{a}$
under the isomorphism $G_{\tau,\alpha,k}$ is a right ideal of $\Petita{S}{b}$,
since $G_{\tau,\alpha,\ell}(ug)=G_{\tau,\alpha,\ell}(u)G_{\tau,\alpha,\ell}(g)$ for all $u\in \Petita{S}{a}$.  Thus we have a one-to-one correspondence of the codes in $\mathbf{C}_a$ and those in $\mathbf{C}_b$.  Since the ring isomorphism $G_{\tau,\alpha,\ell}$ is in particular a $\tau$-semilinear bijective map between the free $S$-modules $\Stsigma/\Stsigma(t^\mathbf{n} -a)$ and $\Stsigma/\Stsigma(t^\mathbf{n} -b)$, it restricts to a  $\tau$-semilinear bijective map between the free $S$-modules $\Stsigma g/\Stsigma(t^\mathbf{n} -a)$ and $\Stsigma G_{\tau,\alpha,\ell}(g)/\Stsigma(t^\mathbf{n} -b)$. Therefore both $C$ and $G_{\tau,\alpha,k}(C)$ must have the same dimension.
\end{proof}

\begin{remark}\label{re:nonfree}
When $S$ is a finite ring we can also argue as follows to get some intuition on why the dimension of a skew constacyclic code will also be  preserved  under isometries $G_{\tau,\alpha,\ell}$. Let $R=\Stsigma$, for short.  Every ring isomorphism $G_{\tau,\alpha,\ell}$ canonically yields a one-to-one correspondence between the left ideal $Rg/R(t^\mathbf{n} -a)$ and the left ideal $RG_{\tau,\alpha,\ell}(g)/R(t^\mathbf{n}-b)$, which means that $|Rg/R(t^\mathbf{n}-a)|=|R G_{\tau,\alpha,\ell}(g)/R(t^\mathbf{n}-b)|$ have the same number of elements. If we  know that both are free submodules of $S^{\mathbf{n}}$, that means $|Rg/R(t^\mathbf{n} -a)|=|S|^k$ and $|G_{\tau,\alpha,\ell}(g)/R(t^\mathbf{n} -b)|=|S|^{k'}$ for some positive integers $k$ and $k'$, so that $k=k'$ and both must have the same dimension.
\end{remark}

\subsection{A surprising consequence}

The surprising scarcity of monomial isomorphisms (as shown here in Theorem~\ref{prop:G})  sheds light on an error in  \cite[Theorem 4]{OuazzouNajmeddineAydin2025}, where four statements are said to be equivalent.  In fact, trying to uncover the truth of this theorem was the impetus for our study.

Set $\field=\mathbb{F}_q=\mathbb{F}_{p^r}$ with primitive element $\xi$, and let  $\sigma(x)=x^{p^s}$. (Note that the order of $\sigma$ is $\frac{r}{\gcd (r,s)}$, not $s$, as stated.)
Then $N_{\de}^\sigma(x)=x^{[\de]_{s}}$ for all $x\in \field$ where
$$
[\de]_{s}=\frac{p^{s\de}-1}{p^{s}-1}=p^{s(\de-1)}+p^{s(\de-2)}+\cdots + p^{s} + 1.
$$
Then the subgroup of $\field^\times$ consisting of all possible values of $N_{\de}^\sigma(\alpha)$ is the one generated by $N_{\de}^\sigma(\xi)=\xi^{[\de]_{s}}$.

 In  \cite{OuazzouNajmeddineAydin2025},  $a,b\in K^\times$ are called \emph{$(\de,\sigma)$-equivalent} if  $\Petita{K}{a}$ and $\Petita{K}{b}$ are $(\de,\sigma)$-equivalent
  and
\emph{$(\de,\sigma)$-isometric} if  $\Petita{K}{a}$ and $\Petita{K}{b}$ are $(\de,\sigma)$-isometric.
Thus briefly, \cite[Theorem 4]{OuazzouNajmeddineAydin2025} asserted the equivalence of:
\begin{description}
\item[(1)] $\Petita{K}{a}$ and $\Petita{K}{b}$ are $(\de,\sigma)$-isometric;
\item[(2) and (3)] $\langle a, N_{\de}^\sigma(\xi)\rangle=\langle b, N_{\de}^\sigma(\xi)\rangle$;
\item[(4)] $\Petita{K}{b^k}$ and $\Petita{K}{a}$ are $(\de,\sigma)$-equivalent, where $k$ is the degree of the isometry in (1).\footnote{The actual conditions on $k$ specified were: $0<k<[\os]_{s}$, $\gcd(k,[\os]_{s})$.}
\end{description}
Note that (1) implies (2) and (3) by Proposition~\ref{L:Gtaualpha} (for monomial homomorphisms of degree $1$) and Theorem~\ref{prop:G} (for monomial homomorphisms of degree $k>1$, when they exist).
However, the converse implication fails, as illustrated in the following counterexample.

\begin{example}\label{Eg:1}
   Choose the parameters $p,r,s,\de$ to satisfy that $(p^{r}-1)|[\de]_{s}$.  For example, one may take $p=3$, $r=2$, $s=1$ and $\de=4$, so that $\os=2$ and $p^r-1=8$ (meaning $\xi^8=1$) and $[\de]_{s}=40$.
    Then $\langle \xi^{[\de]_{s}}\rangle = \{1\}$ and in this case there are very few isomorphisms between these Petit rings.

    We now choose $a$ so that $a$ and $b=a^k$ are two primitive elements of the field that are not Galois conjugates of one another; together with the triviality $N_{\de}^\sigma(\alpha)$ for all $\alpha$, this guarantees that $\tau(a)\notin bN_{\de}^\sigma(\alpha)$.

    For example, let $a=\xi$ and $b=a^k$ for some $2\leq k < p^r-1$ such that $k\neq p^\ell$ for any $\ell$ (ensuring $b\neq \tau(a)$ for any $\tau\in \Gal(\field/\mathbb{F}_p)$) and $\gcd(k,p^{2}-1)=1$ (ensuring $b$ is another primitive element).  Then we have
    $$
    \langle a, \xi^{[\de]_{s}}\rangle = \field^\times = \langle b, \xi^{[\de]_{s}}\rangle
    $$
    so that (2) holds.  Now we claim (1) fails.

    Since $N_{\de}^\sigma(\alpha)=1$ for all $\alpha$, Proposition~\ref{L:Gtaualpha} implies that if there exists a monomial isomorphism of degree one $G_{\tau,\alpha}:\Petita{K}{a}\to\Petita{K}{b}$ we would have $\tau(a)=b$, which is false. Moreover, since $a\notin \field_0$, condition Theorem~\ref{prop:G}[(3)] fails, so there are no monomial homomorphisms of degree $k$ (for any $k>1$).  Thus $\Petita{K}{a}$ and $\Petita{K}{b}$ are not isometric, and (1) fails. \hfill \qedsymbol
\end{example}

The issue is that the condition for isomorphism, that is, that $a\in b\langle \xi^{[\de]_{s}}\rangle$, is strictly weaker than (2), even if one would consider the full class of isomorphisms of the form $G_{\tau,\alpha,k}$.

Further: taking $k=1$ we see that (1) implies (4), but again our results show that the converse is false.  We offer the same counterexample.

\begin{example}
We return to the setting of Example~\ref{Eg:1}, where $b=a^k$.
Let $S=K$ be a field. Then $\id:\Petita{K}{a^k}\to\Petita{K}{b}$ is an isomorphism satisfying (4), but we have already shown that $\Petita{K}{a}$ and $\Petita{K}{b}$ are not isometric, so again (1) fails. \hfill \qedsymbol
\end{example}

These errors stemmed from an expectation of the existence of many monomial homomorphisms of degree other than one.  They
invalidate  the proofs of
   \cite[Theorem 5, Remark 3, Corollary 1,  3, Remark 4]{OuazzouNajmeddineAydin2025} and the proposed classification of skew constacyclic codes up to  $(\de,\sigma)$-isometry in \cite[Sections 5]{OuazzouNajmeddineAydin2025}.

    On the bright side, our results show that in many cases, the notions of $(\de,\sigma)$-isometry and  $(\de,\sigma)$-equivalence coincide (as do our more general notions of isometry and equivalence), so that employing $(\de,\sigma)$-isometry as an equivalence relation to partition skew constacyclic codes will in many cases suprisingly not lead to a finer partition of skew constacyclic codes with the same performance parameters, as we might have all expected initially. This gives considerably more weight to the results obtained in  \cite[Sections 6, 7]{OuazzouNajmeddineAydin2025}.

\subsection{An interesting correction}
Let us now formulate which results do hold and generalize them immediately to the setting of unital commutative rings $\ring$ and allow $\tau$ to be a nontrivial automorphism.

\begin{lemma}\label{le:3.2 general with b in S.II}
Let  $\sigma\in {\rm Aut}(\ring)$ have order $\os$, fixed field $\ring_0$, and let $a,b,\alpha\in S^\times$.
If  there is a Hamming-weight preserving isomorphism
$G:\Petita{S}{a} \to \Petita{S}{b}$ with $G|_S=\tau$, where $\tau\in \Aut(\ring)$ commutes with $\sigma$,
 then
 $$\langle b, N_{\de}^\sigma (S^\times) \rangle= \langle \tau(a), N_{\de}^\sigma (S^\times)\rangle.$$
 In particular, the subgroups in $S^\times/N_{\de}^\sigma(S^\times)$ generated by $\tau(a)$ and $b$ are the same.
\end{lemma}

\begin{proof}
  Suppose there  exists a Hamming-weight preserving ring isomorphism
    given by  $G:\Petita{S}{a} \to \Petita{S}{b}
        $
    such that $G|_S=\tau$ and $G(t)=\alpha t^k$ for some $\alpha\in S^\times$, $0<k\leq \de-1$.

 We now distinguish two cases, employing Theorem~\ref{prop:G}:

 If $a\in S\setminus S_0$ or $a\in S_0^\times$ but $\os\nmid\de$ then $k=1$ and so $\tau(a)=N_{\de}^\sigma(\alpha) b$
by Proposition~\ref{L:Gtaualpha},
hence $\tau(a)$ lies in the subgroup of $S^\times$ that is generated by $b$ and $N_{\de}^\sigma( S^\times) $. This means $\langle \tau(a), N_{\de}^\sigma( S^\times) \rangle \subset \langle b, N_{\de}^\sigma( S^\times ) \rangle$.

If $a,b\in S_0^\times$ and $\os|\de$ then $k>1$ is possible, and we have by Theorem~\ref{prop:G} that
 $k\equiv 1 \mod \os$ and $\tau(a)=N_{\de}^\sigma(\alpha)b^k$.
This implies that $gcd(k,\de)=1$ and thus
 $\tau(a)$ lies in the subgroup of $S^\times$ that is generated by $b$ and $N_{\de}^\sigma( S^\times)$. This means $\langle \tau(a), N_{\de}^\sigma( S^\times) \rangle \subset \langle b, N_{\de}^\sigma( S^\times ) \rangle$.

In both cases, the inverse isomorphism is given by $G_{\tau^{-1},\beta,\ell}$ for some $\beta\in \ring^\times$ and $0<\ell\leq \de-1$, so applying the same arguments yields $\langle \tau^{-1}(b),N_{\de}^\sigma(\ring^\times)\rangle \subset \langle a, N_{\de}^\sigma(\ring^\times)\rangle.$  Since $\tau$ commutes with $\sigma$, $\tau(N_{\de}^\sigma(\ring^\times))=N_{\de}^\sigma(\ring^\times)$ and thus we obtain the desired equality.
\end{proof}

As pointed out in Example~\ref{Eg:1}, the statement in the lemma is not an equivalence, in general, not even over finite fields.

\begin{proposition}\label{t:code}
Let  $\sigma\in {\rm Aut}(\ring)$ have order $\os$, fixed field $\ring_0$, and let $a,b,\alpha\in S^\times$.  Let  $\tau\in \Aut(\ring)$ be any automorphism that commutes with $\sigma$.
\begin{enumerate}[(a)]
\item
 Suppose  $G_{\tau,\alpha,k}:\Petita{S}{a} \to \Petita{S}{b}
 $ is an isomorphism with $k>1$ for some $\alpha\in \ring^\times$.  Then $k\equiv 1 \mod \os$,  $\gcd(k,\de)=1$, and $G_{\tau,\alpha}:\Petita{S}{a} \to \Petita{S}{b^k}
 $ is also an isomorphism.

\item  Then $G_{\tau,\alpha}:\Petita{S}{a} \to \Petita{S}{b}
$ is an isomorphism   if and only if $G_{id,\alpha}:\Petita{S}{\tau(a)} \to \Petita{S}{b}
$ is an isomorphism.

\item
 Also, for any $a\in \mathrm{Fix}(\tau)=\{b\in \ring:\tau(b)=b\}$, we have $G_{\id,\alpha}:\Petita{S}{a} \to \Petita{S}{b}
 $ is an isomorphism if and only if $G_{\tau,\alpha}:\Petita{S}{a} \to \Petita{S}{b}$  is an isomorphism.
 \end{enumerate}
\end{proposition}

\begin{proof}
(a)  Suppose there  exists an isomorphism $$
G_{\tau,\alpha,k}:\Petita{S}{a} \to \Petita{S}{b}
$$ for some $\alpha\in S^\times$, $1< k<\de$.
By Theorem~\ref{prop:G}, since $k>1$, we must have
 $a,b\in S_0^\times$ and $\os|\de$.  Moreover,  by Theorem~\ref{prop:G} we know
 $k\equiv 1 \mod \os$,  $\gcd(k,\de)=1$, and $\tau(a)=N_{\de}^\sigma(\alpha)b^k$.
Therefore by Proposition~\ref{L:Gtaualpha},
$$G_{\tau,\alpha}: \Petita{S}{a} \to \Petita{S}{b^k}
$$ is  an isomorphism.

\noindent (b) By Proposition~\ref{L:Gtaualpha}, $G_{\tau,\alpha}$ is an isomorphism if and only if  $\tau(a)= N_{\de}^\sigma(\alpha)b$. This is equivalent to
$$G_{id,\alpha}: \Petita{S}{\tau(a)} \to \Petita{S}{b}
$$
being a ring isomorphism that restricts to the identity on $\ring_0$.

\noindent (c) is a special case of (b).
\end{proof}

Let us point out that Proposition~\ref{t:code}(b) highlights the advantage of choosing equivalence  over $(\de,\sigma)$-equivalence: when $\tau(a)= N_{\de}^\sigma(\alpha)b$ then the class of skew $(\sigma,a)$-constacyclic codes is equivalent to the class of skew $(\sigma,b)$-constacyclic codes,
but so will be the class of
$(\sigma,\tau(a))$-constacyclic codes for any $\tau$ commuting with $\sigma$.

It also shows that  the class of skew $(\sigma,a)$-constacyclic codes is equivalent to the class of skew $(\sigma,b)$-constacyclic codes via $G_{\tau,\alpha}$, if and only if  the class of skew $(\sigma,\tau(a))$-constacyclic codes is
$(\de,\sigma)$-equivalent to the class of $(\sigma,b)$-constacyclic codes.

Moreover,
(a) implies  that two classes of skew constacyclic codes  that are $(\de,\sigma)$-isometric and related to principal left ideals in $\Petita{S}{a}$,
respectively to principal left ideals in $\Petita{S}{b}$,
will generate two classes of codes in $\Petita{S}{a}$,
respectively, in $\Petita{S}{b}$,
that are $(\de,\sigma)$-equivalent.

In \cite[Theorem 4]{OuazzouNajmeddineAydin2025} the authors claim that the number of $(\de,\sigma)$-isometry classes is equal to the number of divisors of $\gcd([\de]_s,p^r-1)$, but as this was based on the false equivalence (1) $\iff$ (2), the answer was incorrect.

We derive the answer for the case that $\Petita{S}{a}$ is nonassociative, that is, when $(\de,\sigma)$-equivalence and $(\de,\sigma)$-isometry coincide.  Note that the proof does not require that $t^n-a$ is reducible.

\begin{theorem}\label{c:Ouazzoufinite}
    Suppose $K=\mathbb{F}_{p^r}$ and $\sigma(x)=x^{p^s}$ with $s|r$ so that $\os=r/s$ and $\field_0=\mathbb{F}_{p^s}$.
    The number of distinct $(\de,\sigma)$-isometry classes of families of skew $(\sigma,a)$-constacyclic codes arising from nonassociative rings $\Petita{K}{a}$ is $N$ where
    $$
     N= \begin{cases}
     \gcd([\de]_{s},p^r-1)&\text{if $\os\nmid \de$; and}\\
    \left(1-\frac{1}{[\os]_s}\right)\gcd([\de]_{s},p^r-1) & \text{if  $\os \mid\de$.}
    \end{cases}
    $$
    There are additionally $\gcd([\de]_{s},p^r-1)/[\os]_s$ different skew $(\de,\sigma)$-equivalence classes (and thus at most this many $(\de,\sigma)$-isometry classes) of families of skew $(\sigma,a)$-constacyclic codes arising from associative rings $\Petita{K}{a}$, that is, for which $\os\mid \de$ and $a\in \field_0$.
    \end{theorem}

\begin{proof}
We start by recalling that the notions of  $(\de,\sigma)$-equivalence and $(\de,\sigma)$-isometry coincide  when $\Petita{K}{b}$ is a proper nonassociative ring.

    For any $a,b\in K^\times$, the  class of skew $(\sigma,a)$-constacyclic codes is $(\de,\sigma)$-equivalent to the class of skew $(\sigma,b)$-constacyclic codes if and only if
    $$
    b \in aN_{\de}^\sigma(\alpha)
    $$
    for some $\alpha \in K^\times$, that is, if and only if $a$ and $b$ lie in the same coset of the subgroup $N_{\de}^\sigma(K^\times)$.
    Let $\xi$ be a primitive element of $K^\times$; it has order $p^r-1$.  Then $N_{\de}^\sigma(K^\times)$ is generated by $N_{\de}^\sigma(\xi) = \xi^{[\de]_s}=\xi^w$ with $w=\gcd([\de]_s,p^r-1)$, and thus the order of this group is
    $$
    \frac{p^r-1}{w}.
    $$

    First suppose that $\os\nmid \de$, so that all of the rings $\Petita{K}{a}$ are not associative.  In this case, every coset $aN_{\de}^\sigma(K^\times)$ corresponds to a class of $(\sigma,a)$-constacyclic codes such that $\Petita{K}{a}$ is not associative, and there are
    $$
    |K^\times|/|N_{\de}^\sigma(K^\times)|= \gcd([\de]_s,p^r-1)=w
    $$
    such classes.

    Now suppose $\os \mid\de$, so that we are in the nonassociative case if and only if $a,b\in \field\setminus \field_0$. Recall $\os=r/s$.  Then $N_{\de}^\sigma(\mathbb{F}_{p^r})\subset \mathbb{F}_{p^s}$, so that every coset of $N_{\de}^\sigma(\mathbb{F}_{p^r})$ is either entirely contained in $\mathbb{F}_{p^s}$  or entirely disjoint from it (in which case it corresponds to a nonassociative Petit ring).

Note that $\os\mid\de$ implies  $r\mid s\de$ so that  $p^r-1$ divides $p^{s\de}-1$.  Since $[\de]_s=\frac{p^{s\de}-1}{p^s-1}$ we deduce that $[\os]_s=(p^r-1)/(p^s-1)$ divides $w=\gcd([\de]_s,p^r-1)$.  With this in mind we now separate the two cases.

    The total number of cosets contained in $\mathbb{F}_{p^s}$, and thus corresponding to associative rings $\Petit{K}{a}$ with $a\in \mathbb{F}_{p^s}$, is equal to
    $$
    |\mathbb{F}_{p^s}^\times|/|N_{\de}^\sigma(K^\times)| = \frac{p^s-1}{(p^r-1)/w} = \frac{w}{[\os]_s} \in \mathbb{N}.
    $$
In that case, we have shown that the notions of $(\de,\sigma)$-isometry and $(\de,\sigma)$-equivalence need not coincide.  Thus these $(\de,\sigma)$-isometry classes may be bigger (thus fewer in number), yielding the last sentence of the theorem.

Finally, subtracting these $\frac{w}{[\os]_s}$ cosets from the total number of cosets gives the number of $(\de,\sigma)$-isometry classes for nonassociative rings in the case that $\os|\de$, that is, that $N=w-w/[\os]_s$, as required.
\end{proof}

This recovers, in particular, \cite[Theorem 6]{OuazzouNajmeddineAydin2025} as a special case, as this means that the number of  $(\de,\sigma)$-equivalence classes is $\gcd([\de]_{s},p^r-1)$.

\begin{example}
Let $S$ be arbitrary and let $a\in S^\times$.
At one extreme, suppose that $N_{\de}^\sigma(\ring^\times)=\ring^\times$. Then  all of the rings $\Petita{S}{a}$ are $(\de,\sigma)$-equivalent, and  in particular all  are $(\de,\sigma)$-equivalent to $\Petita{S}{1}$
which corresponds to the class of  skew cyclic codes.  Moreover, equivalence and $(\de,\sigma)$-equivalence coincide in this case. When  $\os\nmid \de$, then the notions of equivalence,  isometry and  $(\de,\sigma)$-isometry coincide as well.
Hence for all $a\in S^\times$, the skew $(\sigma,a)$-constacyclic codes of length $\de$ over $S$ are equivalent to skew cyclic codes of length $\de$.

As a special case, let $\field=\mathbb{F}_{p^r}$, $\sigma(x)=x^{p^s}$ with $s|r$  and $\field_0=\mathbb{F}_{p^s}$. If $\gcd([\de]_s,p^r-1)=1$
then $N_{\de}^\sigma(\field^\times)=\field^\times$ so  all $\Petita{K}{a}$ are  $(\de,\sigma)$-equivalent, as was already observed in \cite[Corollary 2]{OuazzouNajmeddineAydin2025} and \cite[Proposition 1]{BoulanouarBatoulBoucher2021}.
\hfill \qedsymbol
\end{example}

\begin{example}\label{e:important}
At the other extreme, when
$\gcd([\de]_s,p^r-1)=p^r-1$ (as in Example~\ref{Eg:1}), then $N_{\de}^\sigma(\field^\times)=\{1\}$ and so no two distinct $a,b\in \field^\times$ will be  $(\de,\sigma)$-equivalent. Thus, there are a plethora of distinct classes of skew constacyclic codes up to $(\de,\sigma)$-equivalence.

That said, there always $r$ choices for $\tau\in {\rm Gal}(\mathbb{F}_{p^r}/\mathbb{F}_{p})$.
We deduce that the set $\{a^{p^{v}}\mid 0\leq v < r\}\subset \mathbb{F}_{q^r}$ is an equivalence class with $r$ elements.  That is, all of the corresponding $(\sigma,a^{p^v})$-constacyclic codes are equivalent and there are in fact  fewer isometry classes overall. \hfill \qedsymbol
    \end{example}

Example~\ref{e:important}  demonstrates that  using our refined notion of equivalence and isometry gives a tighter classification of these classes of codes.

Finally, we derive the following 
powerful equivalence (inspired by the flavour of \cite[Theorem 4]{OuazzouNajmeddineAydin2025}) that holds for  any  commutative ring $S$ and any $\tau$ commuting with $\sigma$.

\begin{theorem}\label{t:Ouazzou}
Suppose that $a,b\in S_0^\times$ and $\os\mid \de$.
 Assume that $\tau$ commutes with $\sigma$.
Then the following statements are equivalent for any integer $1\leq k<\de$:
\begin{enumerate}[(i)]
\item $\Petita{S}{a}$ and $\Petita{S}{b}$
are isometric via $G_{\tau,\alpha,k}$;
\item $\Petita{S}{b^k}$ and $\Petita{S}{a}$
are equivalent via $G_{\tau,\alpha}$, where $k$ satisfies both $k\equiv 1 \mod \os$ and $\gcd(k,\de)=1$.

\item  $\tau(a)=N_{\de}^\sigma(\alpha)  b^k$ where $k\equiv 1 \mod \os$ and  $\gcd(k,\de)=1$.
\end{enumerate}
\end{theorem}

\begin{proof}
The case $k=1$ is trivial, so we may assume $k>1$.

\noindent $(i) \Rightarrow (ii)$: If there exists an isomorphism
  $G_{\tau,\alpha,k}:\Petita{S}{a}\to\Petita{S}{b}$  with $k>1$ for some $\tau$ that commutes with $\sigma$ and some $\alpha\in \ring^\times$,  then
  $k\equiv 1 \mod \os$ and $\gcd(k,\de)=1$ by Theorem~\ref{prop:G} and
  we can construct a degree one isomorphism  $G_{\tau,\alpha}:\Petita{S}{a}\to\Petita{S}{b^k}$ by Proposition~\ref{t:code} (a).

  \noindent $(ii) \Rightarrow (i)$: If we have a degree one isomorphism  $G_{\tau,\alpha}:\Petita{S}{a}\to\Petita{S}{b^k}$
  such that $k\equiv 1 \mod \os$,  $\gcd(k,\de)=1$
  then $\tau(a)=N_{\de}^\sigma(\alpha)  b^k$
  implies that $\tau(a)=(N_{\ring/\ring_0}(\alpha))^{\de/\os}b^k$ holds,  hence $G_{\tau,\alpha,k}:\Petita{S}{a}\to\Petita{S}{b}$ is an isomorphism by Theorem~\ref{prop:G}.

  \noindent $(ii)\iff (iii)$ This is Proposition~\ref{L:Gtaualpha}.
\end{proof}

\subsection{Further examples}

We have established the varying notions of isometry, equivalence, $(\de,\sigma)$-isometry and $(\de,\sigma)$-equivalence.  To illustrate the difference between isometry and  $(\de,\sigma)$-equivalence, consider the following example, where we see the importance of allowing $\tau$ to vary over all $\Aut(K)$.

\begin{example}
Let $K=\F_{16}$ and $\sigma(x)=x^4$ so that $K_0=\F_4$. Then $\os=2$ but $\Aut(K)$ has four elements, generated by $\tau(x)=x^2$.  Let $\de=6$; then $N^\sigma_{\de}(K^\times)=\{1\}$. Therefore two classes ${\bf C}_a$ and ${\bf C}_b$ are $(\de,\sigma)$-equivalent if and only if they are equal \emph{i.e.}, if and only if  $a=b$.
However, applying Proposition~\ref{L:Gtaualpha}, we see that for every $a\in K^\times$ we have the following isometries:
$$
G_{\tau,1}:\Petit{K}{a^2}\to \Petit{K}{a}, \quad G_{\tau^2,1}: \Petit{K}{a^4}\to \Petit{K}{a} \quad \text{and} \quad \Petit{K}{a^8}\to \Petit{K}{a}.
$$
That is, the three classes  ${\bf C}_a$, ${\bf C}_{a^2}$, ${\bf C}_{a^4}$ and ${\bf C}_{a^8}$ are isometric, for every $a\in K^\times$.
\end{example}

Finally, while the focus of this paper has been on establishing the exact conditions under which values $a\in S^\times$ will give rise to isometric families of skew $(\sigma,a)$-constacyclic codes $\mathbf{C}_a$, it is also of further interest to consider isometries between the codes themselves.  We include some simple examples, at the suggestion of the referee, that demonstrate the importance of considering all isometries, not just equivalences, when classifying the codes within such a class $\mathbf{C}_a$; evidently, the phenomena illustrated here extend to all higher dimensions and larger rings $S$.

\begin{example}
Let $K_0=\mathbb{F}_2$, $K=\mathbb{F}_4=\{0,1,\omega,\omega^2\}$ and $\sigma(x)=x^2$. Let $f(t)=t^8-1\in K[t;\sigma]$.
Then $ g_1(t)=t^3+\omega t^2+1$ and $ g_2(t)=t^3+\omega^2 t+1$
are monic right divisors of $f(t)$; we have
 $$t^8-1=(t^5+\omega^2t^4+t^3+\omega t^2+1)g_1(t),\text{ and }
  t^8-1=(t^5+\omega t^3+t^2+\omega^2 t+1) g_2(t).$$
  Therefore the skew $(\sigma,1)$-cyclic codes $C_1$ and $C_2$ that correspond to $g_1(t)$ and $g_2(t)$ are both of length $8$, dimension $5$ and of minimal distance $d=3$.

 By Theorem~\ref{prop:G}, we can see that $G_{\id,1,7}: \Petit{K}{t^8-1} \to \Petit{K}{t^8-1}$ is a ring automorphism.
 We compute that
\[
G_{\id,1,7}(g_1(t))=  g_1(t^7)=t^{21}+\omega t^{14}+1=\omega t^6+t^5+1.
\]
The unique monic polynomial of minimal degree in the ideal generated by $G_{\id,1,7}(g_1(t))$ is the gcd of $\omega t^6+t^5+1$ and $t^8-1$.  Since
\[
  \omega t^6+t^5+1=(\omega t^3+t^2+\omega^2 t+1)g_2(t)
\]
in $K[t;\sigma]$, we infer that the gcd is $g_2(t)$.  Thus the left principal ideals generated by $G_{\id,1,7}(g_1(t))$ and $g_2(t)$ are equal, which shows that the skew $(\sigma,1)$-cyclic code  $C_1$ generated by $g_1(t)$ is isometric to the skew $(\sigma,1)$-cyclic code  $C_2$ generated by $g_2(t)$.

 We now show that these codes 
 $C_1$ and $C_2$ are 
 \emph{not equivalent}, meaning that there does not exist any monomial isomorphism $G_{\tau,\alpha}$ of degree $1$ mapping one to the other.

Let $\tau\in \Aut(K)$ and $\alpha\in K^\times$. Then $\tau(1)=1$ and $\alpha \sigma(\alpha) = \alpha \alpha^2=1$, for all nonzero $\alpha$, since $|K^\times|=3$.  Thus  $N^\sigma_8(\alpha)=(\alpha\sigma(\alpha))^4=1$ so by Theorem~\ref{T:general_isomorphism_theorem2},
$G_{\tau,\alpha}$ is an automorphism of $\Petit{K}{t^8-1}$ for all $\alpha\in\mathbb{F}_4^\times$ and $\tau\in \Aut(K)$.  There are six such pairs; let us evaluate them on $g_1(t)$.

 First observe that for all $\alpha\in \F_4^\times$, we have $N^\sigma_3(\alpha)=\alpha^3\sigma^2(\alpha)=\alpha$ and $N^\sigma_2(\alpha)=\alpha\sigma(\alpha)=\alpha^3=1$.
Thus for all $\tau\in\Aut(K)=\{\mathrm{id},\sigma\}$,
we get
\[
  G_{\tau,\alpha}(g_1(t))=(\alpha t)^3+\tau(\omega)(\alpha t)^2 + \tau(1) = \alpha t^3+\tau(\omega)t^2+1.
\]
We infer that $p_{\tau,\alpha}(t)=t^3+\alpha^{-1}\tau(\omega)t^2+\alpha^{-1}$ is the unique monic polynomial of least degree in the ideal generated by $G_{\tau,\alpha}(g_1(t))$, and this can never be equal to $g_2(t)$:
for every pair $(\tau,\alpha)$, $p_{\tau,\alpha}(t)$ has zero $t$-coefficient whereas $g_2(t)$
has a nonzero $t$-coefficient.  That is, the two ideals can never coincide, and therefore, the two corresponding codes are \emph{isometric but not equivalent}.
\end{example}

We include a final simple example in the setting
of non-chain rings, which gives the flavour of the methods in that setting.

\begin{example}
Let  $S_0=\mathbb{F}_2$, $S=\mathbb{F}_2\times\mathbb{F}_2$,  then $S^\times=\{{\bf 1}:=(1,1)\}$, and $\sigma(x,y)=(y,x)$. Put
${\bf e}=(1,0)$ and ${\bf e'}=(0,1)$.  Let $f(t)=t^6-{\bf 1}\in K[t;\sigma]$. Then $g_1(t)=t^3+t^2+{\bf e}\,t+{\bf 1}$ and $g_2(t)=t^3+{\bf e'}t^2+t+{\bf 1}$ are monic right divisors of $f(t)$, since $f(t)=g_1(t)g_1(t)=(t^3+{\bf e} t^2+t+{\bf1})g_2(t)$. The codes $C_1$ and $C_2$ associated with $g_1(t)$ and $ g_2(t)$ have length $6$ and are free of
rank $3$. Now by Theorem~\ref{prop:G}, since $\de=6$, $\os=2$, and $k=5\equiv 1\mod \os$, $\gcd(k,\de)=1$, we know that $G_{\id,1,5}$ is an
automorphism of  $\Petit{S}{t^6-{\bf 1}}$.  We carefully compute
$$G_{\id,1,5}(g_1(t))=  
t^{15}+t^{10}+{\bf e}t^{5}+1={\bf e} t^5+t^4+t^3+{\bf 1}=({\bf e}\,t^2+t+{\bf 1})g_2(t).$$
Thus the image of the left principal ideal generated by $g_1(t)$ is the left principal ideal generated by $g_2(t)$, that is, these two codes are isometric.

However, since $S^\times=\{1\}$, the only degree-one monomial automorphisms are $\id=G_{\id,1}$ and $G_{\sigma,1}$, and it is easy to check that neither of these gives an \emph{equivalence} between $C_1$ and $C_2$. 
 \end{example}

\section{On non monomial homomorphisms}\label{S:nonmonomial}

Having established the scarcity of monomial homomorphisms, we now show further that if $\ring$ is an integral domain (\emph{e.g.}, a field), then in a large number of cases every homomorphism between proper nonassociative rings of the form $\Petita{S}{b}$
is  monomial of degree one, that is, is of the form $G_{\tau,\alpha}$ for some $\tau\in \Aut(\ring)$ commuting with $\sigma$ and nonzero $\alpha \in \ring$.

We begin by proving a generalization to the polynomial case of Lemma~\ref{L:k=1}.  We assume in this section that $\ring$ has no nonzero zero divisors, that is, it is an integral domain, such as a field.

\begin{theorem}
    \label{P:polyndivm}
    Let $S$ be an integral domain.  Suppose $f(t)\in \Stsigma$ is a monic polynomial of degree $\de$ and $b\in \ring^\times$.  Suppose that $G:\Petit{S}{f} \to \Petita{S}{b}$
        is a nonzero ring homomorphism  whose restriction to $\ring$ is given by some automorphism $\tau$ commuting with $\sigma$, such that $G(t)$ is not a monomial.  Then $\os|\de$ and
    $G(t) = p(t^{\os})t$ for some $p\in \ring[x]$ of degree less than $\frac{\de}{\os}$.
\end{theorem}

\begin{proof}
    Write
    $$
    G(t) = \sum_{i=0}^{\de-1}c_i t^i
    $$
    for some $c_i\in \ring$.
    As in the proof of Lemma~\ref{L:k=1}, which holds for any commutative unital ring $S$, the condition $G(td)=G(\sigma(d)t)$ for all $d\in \ring$ implies that $c_i=0$ for all $i$ such that $i\not\equiv 1 \mod \os$. Therefore $G(t)$ takes the form given in the statement.  It remains to prove that $\os\mid \de$.

    As compared with the monomial case, the proof now is more intricate because each term of $G(t)^s$ is a complex polynomial.
    In the following, we differentiate between the expression $(tp(t^{\os}))^s$ in $\Stsigma$ and the value of $G(t)^s\in \Petita{S}{b}$, that is, after reduction ${\rm mod}_r(t^{\de}-b)$ (whose well-definedness is the key question). Note that when the degree of the polynomial $(tp(t^{\os}))^s$ is less than $\de$, no distinction is needed.

    Let $k=\deg(G(t))$; then $1<k<\de$.  Let $s>1$ be the least value for which $sk\geq \de$; since $k>1$ we have $s<\de$. Then since $G$ is homomorphism, $G(t)=tp(t^{\os})\in \Petita{S}{b}$ must be power associative for all powers less than $\de$ (since this is true in $\Petit{S}{f}$)
    and we have $G(t^s)=G(t)^s$.  Furthermore, since $G(t)^{s-1}=(tp(t^{\os}))^{s-1}$ has degree less than $\de$, all monomials appearing with nonzero coefficient in $(tp(t^{\os}))^s\in \Stsigma$ have degree less than $2\de$.  Thus we may compute $G(t)^s\in \Petita{S}{b}$ by first evaluating the expression $(tp(t^{\os}))^s$ in $\Stsigma$ and then reducing modulo the relation $t^{\de+r}\equiv t^rb\equiv \sigma^r(b)t^r$.

    Write $\mathcal{D}^s$ for the set of all $d\in \mathbb{N}$ that occur as the degree of some monomial appearing in the expansion of $(tp(t^{\os}))^s\in \Stsigma$ and $\mathcal{B}^s$ for the set of all degrees appearing in $G(t)^s\in \Petita{S}{b}$. Then $\mathcal{B}^s \subset \{0,1,\cdots,\de-1\}$ and $\mathcal{D}^s\subset \mathcal{B}^s + (\de+\mathcal{B}^s)$.

    Since each monomial of $G(t)=tp(t^{\os})$ has degree of the form $1+\ell \os$ for some $\ell\geq 0$, it follows that for each $d\in \mathcal{D}^s$ we have  $d\equiv s\mod \os$.
    For each such $d$ that satisfies $d\geq \de$, we obtain a monomial term $\lambda t^{d-\de}$ in the expression for $G(t)^s\in \Petita{S}{b} $, for some $\lambda \in \ring$.

    If $\lambda = 0$, then some cancellation occurred, which necessarily implies that $d-\de=d'$ for some $d'\in \mathcal{D}^s$ satisfying $d'<\de$.  Subtracting $s$ from both sides yields $\de\equiv 0 \mod \os$, or $\os|\de$.

    If $\lambda \neq 0$, then we may write
    $$
    G(t)^s = \lambda t^{d-\de} + q(t) \in \Petita{S}{b}
    $$
    where $q(t)$ contains no monomials of degree $d-\de$.  Let $c\in \ring$.  Then the relation $t^sc=\sigma^s(c)t^s$ holds in the domain of $G$; since $G$ is a homomorphism, this implies $G(t)^s\tau(c)=\tau(\sigma^s(c))G(t)^s$.  This scalar relation holds on each monomial term and so in particular one must have
    $$
    \lambda t^{d-\de} \tau(c) = \tau(\sigma^s(c))\lambda t^{d-\de}.
    $$
    The expression on the left evaluates to  $\lambda \sigma^{d-\de}(\tau(c))t^{d-\de}$, whence, since $\lambda$ is not a  zero divisor, we infer
    $$
    \sigma^{d-\de}(\tau(c))=\tau(\sigma^s(c))
    $$
    for all $c\in \ring$. Since $\tau$ commutes with $\sigma$, this implies $d-\de\equiv s\mod \os$, and so again we deduce  $\de\equiv 0\mod \os$, or $\os|\de$, as required.
\end{proof}

\begin{corollary}\label{c:new}
If   $a,b\in \ring^\times$ and $\os \nmid\de$, then any nonzero homomorphism
 $G:\Petita{S}{a} \to \Petita{S}{b}$  whose restriction to $\ring$ is given by some automorphism $\tau$ commuting with $\sigma$, must be monomial of degree one.
\end{corollary}

\begin{proof}
    The theorem shows that if $\os\nmid\de$, then in this setting, every nonzero ring homomorphism must be such that $G(t)$ is a monomial, that is, $G=G_{\tau,\alpha,k}$ for some $k$.  However since $\os \nmid\de$, Theorem~\ref{prop:G} assures us that we must have $k=1$.  Hence $G$ is monomial of degree one.
\end{proof}

 We now have another more precise result about the nonexistence of nonmonomial homomorphisms, subject to a technical hypothesis \eqref{stareq}.

\begin{theorem}\label{thm:main}
    Suppose $G: \Petit{S}{f}\to \Petita{S}{b} $ is a non monomial homomorphism satisfying $G|_{\ring}=\tau$ for some $\tau\in \Aut(\ring)$ commuting with $\sigma$, as in Theorem~\ref{P:polyndivm}. Let $k=\deg(G(t))$.  Suppose further that
    \begin{equation}\label{stareq}
        \tag{$\star$} \text{if $s>1$ is the least value for which $sk\geq \de$, then $(s+1)k< 2\de$}.
    \end{equation}
    Then $\Petita{S}{b}$
    must be associative, that is, $b\in \ring_0$ and $\os|\de$.
\end{theorem}

For example, hypothesis \eqref{stareq} holds whenever $k<\de/2$.

\begin{proof}
By Theorem~\ref{P:polyndivm}, the degrees of each of the monomials appearing in $G(t)=tp(t^{\os})$ are congruent to $1$ mod $\os$, and $\os|\de$.

We can write  $k=1+\ell \os =\deg(G(t))$.  We have $1\leq k \leq \de-\os+1$.  Since we assume $\os\geq 2$, we have $k\geq 3$ so  $\de\geq 4$.  Let $s>1$ be the least value for which $sk\geq \de$; considering $k=1+\os$ gives the upper bound  $s<\frac{\de}{\os+1}+1\leq \frac{\de}{3}+1$.  Since $s$ is an integer this implies for all $\de\geq 4$ that $s \leq \de-2$.  Therefore  $s+1<\de$ and since $G$ is a homomorphism we have $G(t^{s+1})=G(t)^{s+1}$.    We now compute $G(t)^{s+1}$ in two ways:  as $G(t)^sG(t)$ and as $G(t)G(t)^s$.

We first set some notation.  Let $\mathcal{D}^{s+1}$ be the set of all degrees of monomials appearing in $H(t)=(t(p(t^{\os}))^{s+1}\in \Stsigma$.  Then the maximal element of $\mathcal{D}^{s+1}$ is $d_{max}=(s+1)k$.
Write
$$
G(t) = \sum_{i=0}^{\ell-1}c_it^{1+i\os}
$$
for some $c_i\in \ring$.  For any $d$, let
$$
I_{d} = \{\gamma = (\gamma_1,\gamma_2,\cdots,\gamma_{s+1})\mid c_{\gamma_j}\neq 0, \sum_j (1+\gamma_j \os) = d\}.
$$
These index all the different products of terms of $G(t)$ that give rise to a monomial of degree $d$ in the product $H(t)\in \Stsigma$.  Specifically, given $\gamma\in I_{d}$, the corresponding (noncommutative!) product
is
$$
c_\gamma t^d = \prod_{j=1}^{s+1} (c_{\gamma_j}t^{1+\gamma_j \os}) = \prod_{j=1}^{s+1} \sigma^{j-1}(c_{\gamma_j})t^d
$$
where this is well-defined in $\Stsigma$. If $d<\de$, then $c_\gamma$ is also the coefficient corresponding to this monomial product in $\Petita{S}{b}$,
since products with total degree less than $\de$ are associative.  If $d>\de$, however, then the value of the coefficient will \emph{a priori} depend on the order of associators.

By \eqref{stareq}, we have $(s+1)k<2\de$.  Then $d'=d_{max}-\de=(s+1)k-\de$ is the degree of a monomial appearing in $G(t)^{s+1}\in \Petita{S}{b}$ and we wish to compare its coefficients as computed in the two distinct ways ($G(t)G(t)^s$ and $G(t)^sG(t)$).  The coefficient of $t^{d'}$ in $G(t)^{s+1}$ in either case is the sum of all monomials in $I_{d'}$, plus additional terms that correspond to products in $I_{d_{max}}$ that were reduced by the relation in $\Petita{S}{b}$ to terms of degree $d'$.  By construction, however, $I_{d_{max}}$ is a singleton, being composed only of the products of the top degree terms:
$$
I_{d_{max}} = \{ (\ell,\ell,\cdots,\ell) \} = \{\gamma_{max}\}.
$$
We evaluate the corresponding term in $\Petita{S}{b}$
in two ways as follows.
First note that
$$
(c_\ell t^{k})^s = N_{s}^\sigma(c_\ell)t^{sk}= N_{s}^\sigma(c_\ell)t^{sk-\de}b=N_{s}^\sigma(c_\ell)\sigma^{s}(b)t^{sk-\de}
$$
where we have used right division by $t^{\de}-b$ in $\Petita{S}{b}$
and that $sk-\de\equiv s\mod \os$.
Thus the monomial corresponding to $\gamma_{max}$ in the product $G(t)^sG(t)$ is
$$
c_{left}t^{d'}=(c_\ell t^{k})^s(c_\ell t^{k}) = N_{s}^\sigma(c_\ell)\sigma^{s}(b)t^{sk-\de}c_\ell t^{k}=
\sigma^s(b)N_{s+1}^\sigma(c_\ell)t^{(s+1)k-\de}
$$
whereas the monomial corresponding to $\gamma_{max}$ in the product $G(t)G(t)^s$ is
$$
c_{right}t^{d'}=(c_\ell t^{k})(c_\ell t^{k})^s=c_\ell t^{k}N_{s}^\sigma(c_\ell)\sigma^{s}(b)t^{sk-\de}=
N_{s+1}^\sigma(c_\ell)\sigma^{s+1}(b))t^{(s+1)k-\de}.
$$
Therefore, since $G(t)^sG(t)=G(t)G(t)^s$, the coefficients of $t^{d'}$ in each of these expressions must be equal, yielding
$$
 \sum_{\gamma \in I_{d'}}c_\gamma + c_{left}=\sum_{\gamma \in I_{d'}}c_\gamma+c_{right}.
 $$
As the sums are equal this reduces to the requirement that $c_{left}=c_{right}$, which implies by cancellability of $c_\ell\neq 0$ that $\sigma^s(b)=\sigma^{s+1}(b)$ or $b\in \ring_0$.
\end{proof}

\begin{corollary}
    \label{c:main}
    Suppose $G: \Petit{S}{f}\to \Petita{S}{b} $ is a non monomial isomorphism satisfying $G|_{\ring}=\tau$ for some $\tau\in \Aut(\ring)$ commuting with $\sigma$, and whose degree $k$ satisfies \eqref{stareq}, as in Theorem~\ref{P:polyndivm}.
    Then  $a,b\in \ring_0$ and $\os\mid \de$.
\end{corollary}

On the other hand, in the associative case it is not difficult to produce a homomorphism $G: \Petita{S}{a}\to \Petita{S}{b}$ for which $G(t)$ is a non monomial polynomial.

\begin{example}
    The minimal example is $\os=2,\de=4$ and $G(t)=t+t^3$, with $\tau=\id$.  We choose $\field=\mathbb{F}_{25}$ and $\field_0=\mathbb{F}_5$ with $\sigma(x)=x^5$.  Then $\os|\de$ and $a=b=4\in \field_0$, so $\Petita{K}{a}$ is associative.
Since in addition $G(t)^2=(t+t^3)^2=t^2+2t^4+4t^2=3$, we have $G(t)^4=4$ so $G$ is a well-defined (non-surjective!) polynomial morphism
    $$
    G: \Ktsigma/\Ktsigma(t^4-4)\to \Ktsigma/\Ktsigma(t^4-4).
    $$
\end{example}

The hypothesis \eqref{stareq} is a minor technical condition that should not be necessary.  We present work towards proving the following conjecture in \cite{NevinsPumpluen2026}.

\begin{conjecture}
    \label{conj:nonmonomial}
    Suppose $G: \Petit{S}{f} \to \Petita{S}{b}$ is a non monomial isomorphism satisfying $G|_{\ring}=\tau$ for some $\tau\in \Aut(\ring)$ commuting with $\sigma$.  Then  $a,b\in \ring_0$ and $\os\mid \de$, that is, $\Petita{S}{b}$ must be associative.
\end{conjecture}

As mentioned, having such a complete understanding of all homomorphisms between Petit rings has significant impact to the classification of various associated  error-correcting codes.  In addition to the conjecture above, we identify that it is possible to extend the results of our paper to the case of noncommutative rings.  It would also be interesting to see how our methods can be applied when we
look at homomorphisms between general Petit rings.

\section*{Declarations}

\subsection*{Ethics and consent to participate} N/A
\subsection*{Consent to publish} The authors give consent to publish
\subsection*{Availability of data and materials}  No data were generated or used for this research.
\subsection*{Competing interests}  The authors have no relevant financial or non-financial interests to disclose.
\subsection*{Funding} This paper was written during the second author's stay as a Simons Professor in Residence at the University of Ottawa. She gratefully acknowledges the support of CRM and the Simons Foundation, and thanks the Department of Mathematics and Statistics for its hospitality and its congenial and inspiring atmosphere.  The first author's research is supported by NSERC Discovery Grant RGPIN-2025-05630.
\subsection*{Authors' contributions} Both authors contributed equally to the manuscript.
\subsection*{Acknowledgements} The authors thank  all referees for their valuable suggestions and careful proofreading, which greatly helped to improve the paper.

\newcommand{\etalchar}[1]{$^{#1}$}
\providecommand{\bysame}{\leavevmode\hbox to3em{\hrulefill}\thinspace}
\providecommand{\MR}{\relax\ifhmode\unskip\space\fi MR }
\providecommand{\MRhref}[2]{%
  \href{http://www.ams.org/mathscinet-getitem?mr=#1}{#2}
}
\providecommand{\href}[2]{#2}

\end{document}